\documentclass[a4paper, 10pt]{article}

\usepackage{graphicx,color}
\usepackage{amsmath,epsfig}
\usepackage{times,amsmath,epsfig}
\usepackage{setspace}
\usepackage{amsmath, amsfonts}
\usepackage{amssymb,epsf,alltt}
\usepackage{epsfig}
\usepackage[ruled,vlined]{algorithm2e}
\usepackage{fancyhdr}
\numberwithin{equation}{section}

\newtheorem{theorem}{Theorem}[section]
\newtheorem{lemma}[theorem]{Lemma}

\newtheorem{corollary}[theorem]{Corollary}
\newenvironment{proof}[1][Proof]{\begin{trivlist}
\item[\hskip \labelsep {\bfseries #1}]}{\end{trivlist}}

\newcommand{\REM}[1]{}

\newcommand{\EE}{\mathbb{E}}
\title{
\Huge{\textbf{IBM Research Report}} \\
\vspace*{0.4in}
\Large{\textbf{Multidimensional Balanced Allocation for Multiple Choice \& $(1+\beta)$ Processes}}
\vspace*{0.2in}
}

\small{
\author{
\textbf{Ankur Narang} \\
IBM Research Division \\
IBM India Research Lab \\
Plot - 4, Block - C, ISID Campus, Vasant Kunj \\
New Delhi - 110070, India. \\
annarang@in.ibm.com\\ \\
\textbf{Sourav Dutta} \\
IBM Research Division \\
IBM India Research Lab \\
Plot - 4, Block - C, ISID Campus, Vasant Kunj\\
New Delhi - 110070, India. \\
sodutta3@in.ibm.com \\ \\
\textbf{Souvik Bhattacherjee} \\  
IBM Research Division \\
IBM India Research Lab \\
Plot - 4, Block - C, ISID Campus, Vasant Kunj \\
New Delhi - 110070, India. \\
souvikbh@in.ibm.com\\ \\
}
}

\date{}

\begin{document}

\begin{titlepage}

\maketitle
\thispagestyle{fancy}
\lhead{RI11018, 20 October 2011}
\rhead{Computer Science}
\small{
\noindent\textbf{IBM Research Division} \\
\textbf{Almaden - Austin - Beijing - Delhi - Haifa - T.J. Watson - Tokyo - Zurich} \\

\noindent\textbf{LIMITED DISTRIBUTION NOTICE:} This report has been submitted for publication outside of IBM and will probably be copyrighted is accepted for publication. It has been issued as a Research Report for the early dissemination of its contents. In view of the transfer of copyright to the outside publisher, 
its distribution outside of IBM prior to publication should be limited to peer communications and specific requests. After outside publication, requests should be filled only by reprints or legally obtained copies of the article (e.g., payment of royalties). Copies may be requested from IBM T.J. Watson Research 
Center, Publications, P.O. Box 218, Yorktown Heights, NY 10589 USA (email: reports@us.ibm.com). Some reports are available on the internet at \mbox{http://domino.watson.ibm.com/library/CyberDig.nsf/home}.
}

\end{titlepage}

%%%%
%%%%
%\subjclass{Dummy classification}% mandatory
%\keywords{Algorithm and data structures; Balanced Allocation; Multidimensional Balls and Bins;}
%%%%%%%%%%%%%%%%%%%%%%%%%%%%%%%%%%%%%%%%%%%%%%%%%%%%%%%%%
%%%%%%%%%%%%%%%%
%%%%%
\begin{abstract}
Allocation of balls into bins is a well studied abstraction for load balancing problems. The literature hosts numerous results for sequential (single dimensional) allocation case when $m$ balls are thrown into $n$ bins; such as: for multiple choice paradigm the expected gap between the heaviest bin and the average load is $O(\frac{\log\log(n)}{\log(d)})$~\cite{petra-heavy-case}, $(1+\beta)$ choice paradigm with $O(\frac{\log(n)}{\beta})$ gap~\cite{kunal-beta} as well as for single choice paradigm having $O(\sqrt{\frac{m\log(n)}{n}})$ gap~\cite{mm-thesis}. However, for multidimensional balanced allocations very little is known. Mitzenmacher~\cite{md-mm} proved $O(\log\log(nD))$ gap for the multiple choice strategy and $O(\log(nD))$ gap for single choice paradigm (where $D$ is the total number of dimensions with each ball having exactly $f$ populated dimensions) under the assumption that for each ball $f$ dimensions are uniformly distributed over the $D$ dimensions. In this paper we study the symmetric multiple choice process for both unweighted and weighted balls as well as for both multidimensional and scalar modes. Additionally, we present the results on bounds on gap for the $(1+\beta)$ choice process with multidimensional balls and bins.

In the first part of this paper, we study multidimensional balanced allocations for the symmetric $d$ choice process with $m >> n$ unweighted balls and $n$ bins. We show that for the symmetric $d$ choice process and with $m = O(n)$, the upper bound (assuming uniform distribution of $f$ populated dimensions over $D$ total dimensions) on the gap is $O(\ln\ln(n))$ w.h.p.. This upper bound on the gap is within $D/f$ factor of the lower bound. This is the first such tight result along with detailed analysis for $d$ choice paradigm with multidimensional balls and bins. %Further, we provide upper bounds on the gap when $f$ is a random variable. 
This improves upon the best known prior bound of $O(\log\log(nD))$~\cite{md-mm}. For the general case of $m >> n$ the expected gap is bounded by $O(\ln\ln(n))$. For variable $f$ and non-uniform distribution of the populated dimensions (using analysis for weighted balls), we obtain the upper bound on the expected gap as $O(\log(n))$.
%%%%%%
%%%%%%
\par Further, for the multiple round parallel balls and bins, using symmetric $d$-choice process in multidimensional mode, we show that the gap is also bounded by $O(\log\log(n))$ for $m = O(n)$. The same bound holds for the expected gap when $m >> n$.
%%%%
\par Our analysis also has the following strong implications for the sequential scalar case. For the weighted balls and bins and general case $m >> n$, we show that the upper bound on the expected gap is $O(\log(n))$ (assuming $E[W] = 1$ and second moment of the weight distribution is finite) which improves upon the best prior bound of $n^c$ ($c$ depends on the weight distribution that has finite fourth moment) provided in~\cite{kunal-weighted}. Our analysis also provides a much easier and elegant proof technique (as compared to~\cite{petra-heavy-case}) for the $O(\log\log(n))$ upper bound on the gap for scalar unweighted $m >> n$ balls thrown into $n$ bins using the symmetric multiple choice process. 
%%%%
\par Moreover, we study multidimensional balanced allocations for the $(1+\beta)$ choice process and the multiple ($d$) choice process. We show that for the $(1+\beta)$ choice process and $m=O(n)$ the upper bound (assuming uniform distribution of $f$ populated dimensions over $D$ total dimensions) on the gap is \textbf{$O(\frac{\log(n)}{\beta})$}, which is within $D/f$ factor of the lower bound. For fixed $f$ with non-uniform distribution and for random $f$ with Binomial distribution the expected gap remains \textbf{$O(\frac{\log(n)}{\beta})$} and is independent of the total number of balls thrown, $m$. This is the first such tight result along with detailed analysis for $(1+\beta)$ paradigm with multidimensional balls and bins.
%%%%%%
%%%%%%
%with non-uniform distribution as well as for variable $f$. 
%Moreover, for the multiple choice process with multidimensional balls and bins, we obtain the bound on the gap as $O(\log\log(n))$. 
%
\end{abstract}

\section{Introduction}
\label{sec:intro}

Balls-into-bins processes serve as a useful abstraction for resource balancing tasks in distributed and parallel systems. Assume $m$ balls are to be put sequentially into $n$ bins, where typically the goal is to minimize the load, measured by the number of balls, in the most loaded bin. In the classic \textit{single choice} process each ball is placed in a bin chosen independently and uniformly at random.  For the case of $n$ bins and $m = n$ balls it is well known that the load of the heaviest bin is at most $(1 + o(1)) \frac{\ln (n)}{\ln\ln (n)}$ balls with high probability (w.h.p.). Further, if $m > n\ln (n)$ then the load in the heaviest bin is given by at most $\frac{m}{n} + \sqrt{(\frac{m\log(n)}{n})}$~\cite{simple-analysis-bb}.

An interesting and substantial decrease in the maximum load is achieved by the use of the \textbf{{\it multiple choice}} paradigm (also referred to as \textbf{$d$ choice} paradigm), given as: Let $Greedy(U,d)$ denote the algorithm where each ball is inserted into the lesser loaded among the $d \ge 2$ bins, independently sampled from $U$, where $U$ denotes the uniform distribution over the bins. In a seminal paper Azar et. al.~\cite{bal-alloc-azar} proved that when $m = n$ and the balls are inserted by $Greedy(U,d)$ the heaviest bin has load of $\frac{\ln\ln(n)}{\ln (d)} + \theta(1)$ w.h.p.. The case for $d = 2$ was proved by Karp et.al. in~\cite{karp-pram-sim}, later being generalized by Berenbink et.al.~\cite{petra-heavy-case} to prove the following: 

\begin{theorem}  
 Let $\gamma$ denote a suitable constant. If $m$ balls are allocated into $n$ bins using $Greedy(U,d)$ with $d \ge 2$ then the number of bins with load at least $\frac{m}{n} + i + \gamma$ is at most $n.exp(-d^i)$ with probability at least $1 - 1/n$. (\cite{petra-heavy-case})
\end{theorem}

An immediate corollary is that w.h.p. the heaviest bin has a load of $\frac{m}{n} + \frac{\log\log(n)}{\log(d)} + O(1)$. Thus, the additive gap between the maximum load and the average load is \textit{independent} of the number of balls thrown.
 
The multiple-choice paradigm and balls-and-bins models have several interesting applications. In particular, the two-choice paradigm can be used to reduce the maximum time required to search a hash table. If instead of using a single perfectly random hash function as in a typical hash table implementation (with maximum chain length as $O(\ln(n))$), we use two perfectly random hash functions, then the length of the longest chain reduces to $O(\ln\ln(n))$. In efficient PRAM simulation, the two-choice paradigm helps in reducing the contention (~\cite{pram-sim-meyer}) of processors to access the same memory (DRAM). Further, the two-choice scheme can be advantageous in situations, for example when one hopes to fit one full chain in a single cache line~\cite{lookup-mm}. The multiple-choice approach has also proven useful in online (dynamic) assignment of tasks to servers (disk servers, network servers etc). By using multiple-choice one would get much better load balance across the servers as compared to the single-choice approach.

%Another interesting allocation process is the \textbf{$(1+\beta)$} choice process~\cite{mm-thesis}. Here, with probability, $\beta$, each ball chooses a single bin uniformly and randomly, while, with probability $(1-\beta)$ it chooses the least loaded of $2$ randomly selected bins. If we use $(1+\beta)$ choice then, we would get a larger gap ($O(\log(n)/\beta)$ gap~\cite{kunal-beta} as compared to $O(\log\log(n))$ gap for the two-choice process); but the communication cost to query the load of the servers will be lesser by $(1 + \beta)/2$ factor as compared to the two-choice approach. Thus, the $(1+\beta)$ choice process achieves lower communication cost than the two-choice process at the expense of increased gap. It has been shown (~\cite{kunal-beta}) that this gap is independent of $m$, for arbitrarily large $m$. The $(1+\beta)$ choice process is particularly useful when the communication cost involved in querying the load is too high and/or when too many servers are querying simultaneously. For instance, in a distributed storage system where a frontend server places data items in the backend servers, the $2$-choice scheme requires querying two servers for their load and possibly locking both of them until the data item is placed. These communication overheads become worse in the $d$-choice scheme. Another advantage of considering this process, is that it can model~\cite{mm-thesis} the case where the balls perform the best of $2$ strategy, but some fraction of them are "misinformed" by, say, an unreliable load reporting mechanism, and make the wrong decision.

In many practical problems, the underlying data can be \textit{multidimensional}. This is especially true for parallel data mining and machine learning problems, where the input data has many dimensions such as text search where the distinct words in the document set can be considered as the dimensions and the total number of dimensions equals the size of the vocabulary that could potentially run into millions of words. Because the collection of pages to be indexed is so large, it has to be split among $n$ servers.  When a user makes a query to a front-end machine, the query is sent to all $n$ servers; results are returned to the front-end machine for merging and presentation.  Hence, the time to serve the query is determined by the slowest of the servers, the critical process. The time for each server to process a one-word query is roughly proportional to the number of documents at that server containing the word of interest. Thus, to achieve better efficiency, it is necessary to efficiently split the documents among servers in such way that the number of documents containing a given word is roughly equal.

Further, many application domains such as Telecommunication, Finance and others also involve huge number of dimensions such as genres and sub genres of songs and videos for collaborative filtering~\footnote{http://en.wikipedia.org/wiki/Collaborative\_filtering} type correlational analysis between the users. Here, one would like to predict what type of item (song or video) one user could prefer based on his inferred relationship with other similar users. Due to \textit{massive size} of the multidimensional data in such distributed data mining and machine learning problems, one needs to devise \textit{online load balancing} algorithms. 
While, dimensionality reduction techniques can reduce the total number of dimensions to work on, even then, one needs to handle data with large number of dimensions. Further, this data is highly sparse, i.e. number of filled entries in the \textit{(user * item)} matrix is a small fraction of total possible entries in the matrix. Thus, distributed data mining and machine learning (for example in cloud computing environments), suffer from severe scalability and parallel efficiency issues due to huge load imbalance across the machines in the compute cluster (cloud). Hence, there is a strong need to address load balancing for multidimensional datasets. %Moreover, due to massive size of the data one needs to have low time complexity to achieve load balancing, hence one needs to devise online algorithms for multidimensional datasets. 

\subsection{Probability Distribution for Bin Selection}
\label{subsec:prob}

The $d$-choice scheme can be characterized by a probability vector $p = (p_1, p_2, p_3,...,p_n)$, where $p_i$ denotes the probability a ball falls in the $i^{th}$ most loaded bin. Here, the bins are ordered from the most loaded to the least loaded (ties are broken arbitrarily). Then, $p_1$ denotes the probability that the most loaded bin receives the current ball, $p_2$ denotes the probability that the second bin (in the order) receives the ball and so on. In general, in the $d$-choice scheme, $p_i = (\frac{i}{n})^d - (\frac{i-1}{n})^{d}$. For $d = 1$, $\forall i: p_i = 1/n$, while for $d > 1$, $p_i > p_j$ for $i > j$. Thus, for $d > 1$, the process has bias towards the lighter bins. This biasing leads to an overall lower gap for $d > 1$ choice as compared to single choice ($d = 1$) process.
%%%%%%%
%One can see that in case of the $(1+\beta)$ choice process, $p_i = \beta\frac{(2i - 1)}{n^2} + \frac{(1-\beta)}{n}$. 
%%%%%%%
\par In this paper, we consider the multidimensional variant of the balls and bins problem. One multidimensional variant, proposed by~\cite{md-mm} is as follows: Consider throwing $m$ balls into $n$ bins, where each ball is a uniform $D$-dimensional $0$-$1$ vector of weight $f$. Here, each ball has exactly $f$ non-zero entries chosen uniformly among all $\binom{D}{f}$ possibilities (Fig.~\ref{fig:md-balls-bins} (in the Appendix~\ref{app:fig})). The average load in each dimension for each bin is given as $mf/nD$. Let $l(a,b)$ be the load in the dimension $a$ for the $b^{th}$ bin. The gap in a dimension (across the bins) is given by $gap(a) = \max_b l(a,b) - avg(a)$, where $avg(a)$ is the average load in the dimension $a$. The maximum gap across all the dimensions, $\max_a gap(a)$, then determines the load balance across all the bins and the dimensions. %The \textit{maximum dimensional load} (as opposed to the maximum load over all the bins) is \textit{max_load}$ = \max_{a,b} l(a,b)$ and the gap per dimension is given by $$. Intuitively, consider the maximum load in each dimension and then take the maximum over all dimensions. 
Thus, for the multidimensional balanced allocation problem, the objective is to minimize the maximum gap (across any dimension). We refer to the multidimensional ball as \textit{md-ball} and the multidimensional bin as \textit{md-bin}.

In another variation of multidimensional balanced allocation the constraint of uniform distribution for populated entries is removed. Here again, each ball is a $D$ dimensional $0$-$1$ vector and each ball has exactly $f$ populated dimensions, but these populated dimensions can have an arbitrary distribution. In the third variation that is most general of the three, the number of populated dimensions, $f$, may be different across the balls, where $f$ then is a random variable with an appropriate distribution.

Mitzenmacher et.al. in ~\cite{md-mm} addressed both the single choice and $d$-choice paradigm for multidimensional balls and bins under the assumption that balls are uniform $D$-dimensional $(0,1)$ vectors, where each ball has exactly $f$ populated dimensions. They show that the gap for multidimensional balls and bins, using the two-choice process, is bounded by $O(\log\log(nD))$. However, this result is not tight and assumes that $f$ is $polylog(n)$. Due to arbitrary number of dimensions and the resulting discrepancy across the dimensions along with the general case of $m >> n$, the balanced allocations for multidimensional balls and bins is a challenging problem. In this paper, we compute bounds on the gap for the symmetric $d$-choice process for multidimensional balls and bins for both sequential and parallel scenarios.

\subsection{Summary of Key Results \& Techniques}
\label{subsec:key}
%%%%%%
\par We present detailed analysis for the online sequential and parallel multidimensional balls and bins using the symmetric $d$-choice process and show that for $n$ bins and $m = O(n)$ balls, the gap (assuming that exactly $f$ populated dimensions are uniformly distributed over $D$ per ball) achieved is $O(\ln\ln(n))$. We establish the first ever known bound for $d$-choice process and also show that this bound is tight (within $D/f$ factor) by providing the lower bound as well. This improves upon the best prior bound of $O(\log\log(nD))$~\cite{md-mm}. For the general case of $m >> n$, the upper bound on the gap is $O((\frac{mf}{nD})^{(1/2+\zeta)}\ln\ln(n))$ w.h.p., while the expected gap is still $O(\ln\ln(n))$. For non-uniform distribution with fixed $f$ and for variable $f$ with binomial distribution, we show that the expected gap is still independent of $m$. %Upper bound on the gap is also derived for variable $f$, assuming it is a binomial distribution. 
\par In order to arrive at these results, a novel generic potential function based approach along with \textit{sum load} across the dimensions per bin is used. This is much more challenging than the analysis presented by~\cite{kunal-beta} for $(1+\beta)$-choice process, as we obtain a much tighter bound of $O(\ln\ln(n))$ (as compared to $O(\frac{\log(n)}{\beta})$ in~\cite{kunal-beta}). This requires a novel potential function as well as a much tighter analysis in each lemma to ensure that the expected value of the potential function is less than $O(\ln(n))$ at all time $t$ and satisfies the super-martingale property. 
%%%%
\par For parallel multidimensional balls and bins with multiple rounds using the $d$ choice process, we show the upper bound on the gap as $O(\log\log(n))$ for $m = O(n)$ balls; and extend this bound on the gap to the general case of $m >> n$. This is tighter than the $O(\log\log(nD))$ that can be obtained using the analysis similar to~\cite{md-mm}.
%%%%
\par For the weighted and heavy case ($m >> n$) using symmetric multiple choice sequential process for scalar balls, we prove an upper bound of $O(W^*\log(n))$ (where $W^*$ is the expected weight of the distribution), which improves upon the best prior bound of $O(n^c)$~\cite{kunal-weighted}. Our analysis technique also provides an alternate proof for the symmetric $d$-choice process with scalar unweighted $m >> n$ balls into $n$ bins, that is simpler and elegant as compared to~\cite{petra-heavy-case}.  
%%%%%
\par Further, we present the analysis for bounds on the gap for $(1+\beta)$ choice multidimensional process and prove that for $m = O(n)$ the upper bound on the gap is $O(\frac{\log(n)}{\beta})$ w.h.p. for uniform distribution of $f$ dimensions over the $D$ dimensions. For non-uniform distribution with fixed $f$ and also for variable $f$ the expected gap is $O(\frac{\log(n)}{\beta})$ which is independent of $m$. Table~\ref{table-comp} summarizes the comparison between our upper bounds and the best known prior bounds, with key results highlighted.%Further, using witness tree based analysis we obtain an upper bound for the gap for $d$-choice process as $O(\log\log(n))$, which improves the best known prior~\cite{md-mm} upper bound of $O(\log\log(nD))$.
%%The experimental results (refer Appendix) on randomly generated data conform to the theoretically established bounds on the gap.
%%%%%%%%%%
%%%%%%%%%% Table of Results Comparison %%%%%%%%%%%
%%%%%%%%%%
%%%%%%%%%%
\begin{table}[h]
\begin{center}
\begin{tabular}{|l|l|r|l|}
\hline
\textbf{Process: $d$-choice} & \textbf{Best Prior Bound} & \textbf{Our Bound}\\
\hline
Multidim, Fixed-f,$m=O(n)$ & $O(\log\log(nD))$~\cite{md-mm} & $O(\ln\ln(n))$\\
\hline
\textbf{Multidim, Fixed-f},$m>>n$ & None & $O(\ln\ln(n))$ (expected)\\
\hline
\textbf{Multidim, Var-f} & None & $O(\log(n))$ (expected)\\
\hline
\textbf{Weighted Scalar}, $m>>n$ & $O(n^c)$ (short memory via coupling~\cite{kunal-weighted}) & $O(\log(n))$\\
\hline
Unweighted Scalar, $m>>n$ & $O(\ln\ln(n))$ (using layered induction & $O(\ln\ln(n))$\\
 & and short memory~\cite{petra-heavy-case}) & (using potential function)\\
\hline
Parallel Multidim, $m=O(n)$ & $O(\log\log(nD))$ (adaptation of~\cite{md-mm}) & $O(\ln\ln(n))$\\
\hline
\textbf{Parallel Multidim}, $m>>n$ & None & $O(\ln\ln(n))$ (expected)\\
\hline
Parallel Scalar & $O(\ln\ln(n))$ (for $m=O(n)$~\cite{adler95}) & $O(\ln\ln(n))$ (for $m>>n$)\\
\hline
\hline
\textbf{Process: $(1+\beta)$-choice} & \textbf{Best Prior Bound} & \textbf{Our Bound}\\
\hline
Multidim, Fixed-f,$m=O(n)$ & None & $O(\frac{\log(n)}{\beta})$\\
\hline
\textbf{Multidim, Fixed-f},$m>>n$ & None & $O(\frac{\log(n)}{\beta})$ (expected)\\
\hline
\textbf{Multidim, Var-f} & None & $O(\frac{\log(n)}{\beta})$ (expected)\\
\hline
\end{tabular}
\caption{Upper Bound Comparison for $d$-choice and $(1+\beta)$ Process}
\label{table-comp}
\end{center}
\end{table}
%%%%%%
%%%%%%
%%%%%%
\section{Related Work}
\label{sec:related}

Balls into bins is a well studied abstraction for load balancing problems. Numerous results are known for sequential (single dimensional) allocation case when $m$ balls are thrown into $n$ bins; such as: for multiple choice paradigm the expected gap between the heaviest bin and the average load is $O(\frac{\log\log(n)}{\log(d)})$~\cite{petra-heavy-case}, $(1+\beta)$ choice paradigm with $O(\frac{\log(n)}{\beta})$ gap~\cite{kunal-beta} as well as for single choice paradigm having $O(\sqrt{\frac{m\log(n)}{n}})$ gap~\cite{mm-thesis}.~\cite{bal-alloc-azar} showed that the bound of $O(\frac{\log\log(n)}{\log(d)})$ for the symmetric $d$ choice process is stochastically optimal, i.e. any other greedy approach using the placement information of the previous balls to place the current ball majorizes to their approach. However, if the alternatives are drawn from different groups then different rules for tie breaking result in different allocations.~\cite{vocking-tree} presents such an \textit{asymmetric} strategy and using witness tree based analysis proves that this leads to improvement in load balance to $O(\frac{\log \log(n)}{d\log(\phi_d)})$ w.h.p. where, $\phi_2$ is the golden ratio and $\phi_d$ is a simple generalization.
%%%%%
%%%%%
\par The multiple choice and in particular the two-choice paradigm and balls-and-bins models have several interesting applications.
In particular, the two-choice paradigm can be used to reduce the maximum search time in a hash table. Instead of using a single perfectly random hash function as in a typical hash table implementation (with maximum chain length as $O(\ln(n))$), if we use two perfectly random hash functions, then the length of the longest chain reduces to $O(\ln\ln(n))$. In the latter case, when inserting a key, we apply both hash functions to determine the two possible table entries where the key can be inserted. Then, of the two possible entries, we add the key to the shorter of the two chains. To search for an element, we have to search through the chains at the two entries given by both hash functions. If $n$ keys are sequentially inserted into the table, the length of the longest chain is $O(\log\log n)$ with high probability, implying that the maximum time needed to search the hash table is $O(\log\log n)$ with high probability. 
Further, the two-choice scheme can be advantageous in situations for example, when one hopes to fit one full chain in a single cache line~\cite{lookup-mm}. The two-choice approach has also proven useful in online (dynamic) assignment of tasks to servers (disk servers, network servers etc). By using two-choice one would get much better load balance across the servers as compared to the single-choice approach. If we use $(1+\beta)$ choice then, we would get around $\log(n)$ gap (as compared to $O(\log\log(n)$ gap for the two-choice) but the communication cost to query the load of the servers will be lesser by $(1 + \beta)/2$ factor as compared to the two-choice approach.

Cole et al.~\cite{routing-cole} show that the two-choice paradigm can be applied effectively in a different context, namely, that of routing virtual circuits in interconnection networks with low congestion. They show how to incorporate the two-choice approach to a well-studied paradigm due to Valiant for routing virtual circuits to achieve significantly lower congestion.

Kunal et.al.~\cite{kunal-beta} present that for online sequential $(1+\beta)$ choice process with $n$ bins and $m >> n$ balls, a tight gap of $O(\frac{\log(n)}{\beta})$ can be obtained. They use a potential function based technique and further use a majorization argument to generalize their result. We present a novel generic potential function based approach with \textit{sum load} function across all dimensions of a bin for multidimensional balls and bins and obtain tight bounds on the gap for the $d$-choice process for both sequential and parallel scenarios. Our analysis is much more challenging than~\cite{kunal-beta} since we prove a tighter bound that requires a much tighter analysis in each lemma to prove that the expected value of potential function is less than $O(\ln(n))$ at all time $t$. Further, the lower and upper bounds for the $(1+\beta)$ choice process with multidimensional balls and bins have also been provided in this paper.

\par Mitzenmacher et.al. in~\cite{md-mm} address both the single choice and $d$-choice paradigm for multidimensional balls and bins under the assumption that balls are uniform $D$-dimensional $(0,1)$ vectors, where each ball has exactly $f$ populated dimensions. They show that the gap for multidimensional balls and bins, using the two-choice process, is bounded by $O(\log\log(nD))$. We provide better bound on the gap ($O(\log\log(n))$) and also provide the bound for the general case of $m >> n$. Further, while~\cite{md-mm} assumes that $f$ is $polylog(n)$ we don't make any such assumptions. For the multiple round multidimensional parallel balls and bins process where in each round, each bin accepts at max only a single ball, one can use layered induction based proof~\cite{md-mm} to get a similar bound on the gap as $O(\log\log(nD))$. Using our novel potential function based analysis we show a tighter upper bound of $O(\log\log(n))$. The bound for the general case of $m >> n$ is also provided. 
%%%%%%
\par Berenbrink et.al.~\cite{petra-heavy-case} prove an upper bound of $O(\log\log(n))$ for the general case of $m >> n$ balls and $n$ bins using a sophisticated analysis involving two main steps. In the first step, they show that when the number of balls is polynomially bounded by the number of bins the gap can be bounded by $O(\ln\ln(n))$, using the concept of layered induction and some additional tricks. In particular, they consider the entire distribution of the bins in the analysis (while in typical $m = O(n)$ case the bins with load smaller than the average could be ignored). In the second step, they extend this result to general $m >> n$ case, by showing that the multiple-choice processes are fundamentally different from the classical single-choice process in that they have \textit{short memory}. This property states that given some initial configuration with gap $\Delta$, after adding $poly(n)$ more balls the initial configuration is \textit{forgotten}. The proof of the short memory property is done by analyzing the mixing time of the underlying Markov chain describing the load distribution of the bins. The study of the mixing time is via a new variant of the coupling method (called \textit{neighboring coupling}). We prove the same result on the gap ($O(\log\log(n))$) for the symmetric $d$ choice process with $m >> n$ but by using a much simpler and elegant potential function based approach.
%%%%%
\par Kunal et.al.~\cite{kunal-weighted} prove that for weighted balls (weight distribution with finite fourth moment) and $m >> n$, the expected gap is independent of the number of balls and is less than $n^c$, where $c$ depends on the weight distribution. They first prove the weak gap theorem which says that w.h.p $Gap(t) < t^{2/3}$. Since in the weighted case the $d$ choice process is not dominated by the one choice process, they prove the weak gap theorem via a potential function argument. Then, the \textit{short memory theorem} is proved.  While in~\cite{petra-heavy-case} the short memory theorem is proven via coupling,~\cite{kunal-weighted} uses similar coupling arguments but defines a different distance function and use a sophisticated argument to show that the coupling converges. ~\cite{kunal-weighted} also presents a reduction from the real-weighted case to the integer-weighted case. We present the results for weighted case (with integer and real weights and weight distribution with finite second moment) using an elegant and much simpler potential function based argument and show that the gap for arbitrary $m >> n$ is bounded by $O(W^*\log(n))$, where $W^*$ is the expected weight of the distribution.
%%%%%
Adler et.al.~\cite{adler95} consider parallel balls and bins with multiple rounds. They present analysis for $O(\frac{\log\log(n)}{\log(d)})$ bound on the gap (for $m = O(n)$) using $O(\frac{\log\log(n)}{\log(d)} + O(d))$ rounds of communication. We generalize this result to the case of parallel multidimensional balls and bins and arbitrary $m >> n$ balls with similar bound on the gap.
%%%%%
%%%%%
%%%%%
\newcommand{\Xten}{{\sf X10}}
\newcommand{\gasnet}{{\sf GASNet}}

\section{Symmetric $d$-choice Process}
%%%%%
\par In this section, we present various results on the bounds on the gap using the symmetric $d$-choice process including unweighted sequential and parallel multidimensional balls and bins and the sequential weighted scalar case.
%%%%
%%%%
\subsection{Markov Chain Specification}\label{sec:markov-chain}
\label{sec:markov}

As mentioned earlier, a balls-and-bins process can be characterized by a probability distribution vector $(p_1, p_2, p_3,...p_n)$, where, $p_i$ is the probability a ball is placed in the $i^{th}$ most loaded multidimensional bin. Let $x_i^d(t)$ be the random variable, that denotes the \textit{weight in dimension $d$ for bin $i$} and is equal to the load of the $d^{th}$ dimension of the $i^{th}$ bin minus the average load in dimension $d$. So, $\sum_{i=1}^{n} x_i^d(t) = 0, \forall d \in [1..D]$. Each md-ball has $f$ populated dimensions, where $f$ could be constant across the balls or a random variable with a given distribution. Let, $s_i(t)$ denote the sum of the loads (minus corresponding dimension averages) across all $D$ dimensions for the bin $i$ at time $t$, expressed as $s_i(t) = \sum_{d=1}^D x_i^d$. It is assumed that bins are sorted by $s_i(t)$. So, $s_i \ge s_{i+1} \forall i \in [1..n-1]$. The process defines a Markov chain over the matrices, $x(t)$ as follows:

\begin{itemize}
\item Sample $j \in_p [n]$.
\item Set $r_i = s_i(t) + f(1 - 1/n)$, for $i = j$. Since, an md-ball has $f$ non-zero entries , so each of these $f$ dimensions in the bin, $i$, will be incremented by $1 - 1/n$.
\item Set $r_i = s_i(t) - f/n$, for $i \ne j$. Since, an md-ball has $f$ non-zero entries, so the each of the corresponding $f$ dimensions in the bin, $i$, will be decremented by $1/n$. This ensures that for each dimension the sum across all the bins is $0$.
\item Obtain $s(t+1)$ by sorting $r(t)$.
\end{itemize}

Fig.~\ref{fig:md-balls-bins} (in the Appendix~\ref{app:fig}) illustrates a multidimensional balls and bins scenario. The bounds on the gap will be proven for a family of probability distribution vectors $p$. As mentioned earlier, the md-bins are sorted based on their total dimensional load, i.e. sum of the weights across all dimensions for each bin ($s_i$ for bin $i$). 
%We make the following assumptions:
%\begin{itemize}
%
%\item %% Ass(1) %%%
%  $\forall i \in [1,n-1], p_i \le p_{i+1}$
%  This assumption states that the allocation rule is no worse than the $1$-choice scheme.
%
%\item %% Ass(2) %%%
  
%%%%
%\begin{equation}\label{ass-2}
%  p_{(n\gamma_3)} \le \frac{(1 - \theta\epsilon)}{n},\quad \text{and},\quad p_{(n\gamma_4)} \ge \frac{(1 + \theta\epsilon)}{n}
%\end{equation}
%  This assumption states that the allocation rule strictly prefers the least loaded $\gamma_3$ fraction of the $n$ bins over the most loaded $(1 - \gamma_4)$ fraction. 
%\end{itemize}

%These assumptions imply that for some constants, $\gamma_1$ and $\gamma_2$, where, $0 < \gamma_1 < \gamma_3 < 1/2 < \gamma_4 < \gamma_2 < 1$ and $\theta\gamma_1 = 1$; $\gamma_1 + \gamma_2 = 1$, we have the following:\\ $\sum_{i \ge (n\gamma_2)} p_i \ge (\gamma_1 + \epsilon)$ and $\sum_{i \le (n\gamma_1)} p_i \le (\gamma_1 - \epsilon)$. This will be useful in the proof. Note that the $(1+\beta)$ choice process satisfies these assumptions for $\epsilon = \beta(1 -2\gamma_3)/\theta$, since $p_{(n\gamma_3)} \le (1 - \beta)/n + 2(n\gamma_3 - 1)\beta/n^2 \le (1 - \beta(1-2\gamma_3))/n$, and similarly $p_{(n\gamma_4)} \ge (1 + \beta(2\gamma_4 - 1))/n$.

In the remaining analysis, we assume that when an md-ball arrives, then the selection of the bins is based on $s_i$, i.e. total sum of weights across all dimensions for the randomly selected bins (Fig.~\ref{fig:md-balls-bins} in Appendix~\ref{app:fig}). In particular, for the $d=2$ choice process, when $d$ bins are randomly selected, the md-ball (with $f$ non-zero entries) is assigned to the md-bin with the lowest $s_i$. Using this selection mechanism, we prove the upper and lower bounds on the gap obtained for the $d$ choice process. Note that, this is a different allocation mechanism than that considered in~\cite{md-mm} where the \textit{max} criteria is used over the restricted set of $f$ populated dimensions in the current md-ball. Further, we prove upper bound for the case when $m >> n$, while~\cite{md-mm} considers $m = O(n)$ case. The proofs below hold for even the case when $d > 2$, though we consider the case for $d=2$ for sake of clarity.

\subsection{Upper Bound On Gap for Unweighted Case}\label{sec:unweighted}
\par Let there be some constants, $\epsilon > 0$, $\theta > 1$, $\gamma_1$ and $\gamma_2$, $\gamma_3$, $\gamma_4$, where, $0 < \gamma_1 < \gamma_3 < 1/2 < \gamma_4 < \gamma_2 < 1$ and $\theta\gamma_1 = 1$; $\gamma_1 + \gamma_2 = 1$, and $\gamma_3 + \gamma_4 = 1$. Since we consider the $d$-choice process, the probability of selecting the bins has strong bias in favor of the lightly loaded bins. For $d=2$, this results in the following:
%%%%
\begin{equation}\nonumber
\begin{aligned}
  p_{(n\gamma_3)} & \le \frac{2n\gamma_3 - 1}{n^2}\\
  p_{(n\gamma_4)} & \ge \frac{2n\gamma_4 - 1}{n^2}
\end{aligned}
\end{equation}
%%%%
This implies that:$\sum_{i \ge (n\gamma_2)} p_i \ge (1 - \gamma_2^2)$ and $\sum_{i \le (n\gamma_1)} p_i \le \gamma_1^2$. We assume that $\epsilon \le 1/4$. Further, let $\alpha = \epsilon/2f$. In the analysis below, we assume each md-ball has exactly $f$ populated dimensions ($f$ constant, fixed-f case). This is similar to the \textit{unweighted case} with scalar balls.
%%%%%%
\par The md-bins can be arranged in a partial order, according to their $s_i$ values. Define an \textit{equi-load} group (say $p$) as a set of bins with the same $s_i$ value. Define the potential of an equi-load group ($p$) as $\Phi(G_p) = \sum_{k=0}^{|G_p|-1} \frac{e^{\alpha.s_k}}{p_0 + k}$, where $p_0$ is the beginning index for the group,  $|G_p|$ is the size of the $p^{th}$ group and $e^{\alpha.s_k} = e^{\alpha.s_{k+1}}, \forall k \in [p_0..(p_0 + |G_p|-2)]$. The $n$ bins are partitioned into disjoint equi-load groups (total $|G|$ groups), i.e. each bin is assigned to only a single equi-load group. The group structure defined here helps in characterizing the change in index of the bin that gets the ball (after sorting).
%%%%
\par Similarly, define another potential function for an equi-load group as, $\Psi(G_p) = \sum_{k=0}^{|G_p|-1} \frac{e^{-\alpha.s_k}}{p_0 + k}$. Now, define the following potential functions over all the groups:
\begin{equation}
\begin{aligned}
     \Phi(t) & = \Phi(s(t)) = \sum_{p=1}^{|G|} \Phi(G_p)\\
     \Psi(t) & = \Psi(s(t)) = \sum_{p=1}^{|G|} \Psi(G_p)\\     
     \Gamma(t) & = \Gamma(s(t)) = \Phi(t) + \Psi(t)\\
               & =  \sum_{i=1}^{n} [ \frac{e^{\alpha.s_i}}{i} + \frac{e^{-\alpha.s_i}}{i} ]\\   
\end{aligned}
\end{equation}
  where, $s_i(t) = \sum_{d=1}^{D} x_i^d(t)$

In the beginning, each dimension for each bin has $0$ weight, thus $s_i = 0, \forall  i$ and hence, $\Gamma(0) = 2\ln(n)$. We show that if $\Gamma(x(t)) \ge a\ln(n)$ for some $a > 0$, then $\mathbb{E}[\Gamma(t+1) | x(t)] \le (1 - \frac{\epsilon}{8n(1 + \epsilon\gamma_1)}) * \Gamma(t)$. This helps in demonstrating that for every given $t$, $\mathbb{E}[\Gamma(t)] \in O(\ln(n))$. This implies that the maximum gap is $O(\ln\ln(n))$ w.h.p.

First, consider the change in $\Phi(t)$ (also refers to $\Phi$ by default) and $\Psi(t)$ (also refers to $\Psi$ by default) separately when a ball is thrown with the given probability distribution. 
\begin{lemma}\label{lemma-phi-base}
  When an md-ball is thrown into an md-bin, the following inequality holds:
\begin{equation}
   \EE[ \Phi(t+1) - \Phi(t) | x(t) ] \le \sum_{i=1}^{n} [ p_i * (\alpha.f + (\alpha.f)^2)  - \alpha.f/n ]. e^{\alpha.s_i} 
\end{equation}
\end{lemma}
\begin{proof}
  Let $\Delta_i$ be the expected change in $\Phi$ if the ball is put in bin, $i$. So, $r_i(t+1) = s_i + f(1-1/n)$; and for $j \ne i$, $r_j(t+1) = s_j(t) - f/n$. The new values i.e. $s(t+1)$ are obtained by sorting $r(t+1)$ and $\Phi(s) = \Phi(r)$. When, an md- ball is committed to bin $i$, then it moves to the end of the previous equi-load group or it creates a new equi-load group and hence can be located at index $p_0$ (beginning location of its prior group) in the new sorted order of the bins. Thus, the expected contribution of bin, $i$, to $\Delta_i$ is given as follows:
\begin{equation}\nonumber
\begin{aligned}
   \EE[ \frac{e^{\alpha.(s_i + f(1-1/n))}}{p_0} ] - \frac{e^{\alpha.s_i}}{p_0}
   & = \frac{e^{\alpha.s_i}}{p_0} [ e^{\alpha.f(1-1/n)} - 1 ]
\end{aligned}
\end{equation}
 Similarly, the expected contribution of bin, $j$ ($j \ne i$) to $\Delta_i$ is given as:
\begin{equation}\nonumber
\begin{aligned}
   \EE[ \frac{e^{\alpha.(s_j - f/n)}}{j} ] - \frac{e^{\alpha.s_j}}{j}
   & = \frac{e^{\alpha.s_j}}{j} [ e^{-\alpha.f/n} - 1 ]
\end{aligned}
\end{equation}
 Therefore, $\Delta_i$ is given as follows:
\begin{equation}\nonumber
\begin{aligned}
   \Delta_i & = \Phi_i [ e^{\alpha.f(1-1/n)} - 1] \, + \,
              \sum_{j \ne i} \Phi_j(e^{-\alpha.f/n} - 1)\\
   & = \Phi_ie^{-\alpha f/n} (e^{\alpha.f} - 1) + (e^{-\alpha.f/n} - 1).\Phi
\end{aligned}
\end{equation}
  Thus, we get the overall expected change in $\Phi$ as follows:
\begin{equation}
\begin{aligned}
   \EE[ \Phi(t+1) - \Phi(t) | x(t) ] & = \sum_{i=1}^{n} p_i * \Delta_i\\
   & = \sum_{i=1}^{n} p_i * [ \Phi_ie^{-\alpha f/n} (e^{\alpha.f} - 1) + (e^{-\alpha.f/n} - 1).\Phi ]\\
   & = \sum_{i=1}^{n} p_i * e^{-\alpha f/n}\Phi_i (e^{\alpha.f} - 1) +                        (e^{-\alpha.f/n} - 1).p_i.\Phi\\
   & = \sum_{i=1}^{n} [ p_i * e^{-\alpha.f/n} (e^{\alpha.f} - 1) + (e^{-\alpha.f/n} - 1) ]\Phi_i
\end{aligned}
\end{equation}
 Now, $e^{(-\alpha.f/n)} * (e^{\alpha.f} - 1)$ can be approximated as follows:
\begin{equation}\nonumber
\begin{aligned}
   e^{(-\alpha.f/n)}.(e^{\alpha.f} - 1) & \le
   (1 - \alpha.f/n + (\alpha.f/n)^2) * (1 + \alpha.f + (\alpha.f)^2 - 1)\\
   & \sim \alpha.f + (\alpha.f)^2 + O((\alpha.f)^2/n)\\   
   e^{(-\alpha.f/n)}.(e^{\alpha.f} - 1) & \lessapprox (\alpha.f + (\alpha.f)^2)
\end{aligned}
\end{equation}
Above, since, $(\alpha.f)^2/n$ is very small for large $n$, we have ignored the small terms. Similarly, $(e^{-\alpha.f/n} - 1) \lessapprox -\alpha.f/n$
Hence, the expected change in $\Phi$ can be given by:
\begin{equation}\label{eq-phi-base}
\begin{aligned}
   \EE[ \Phi(t+1) - \Phi(t) | x(t) ]
   \le \sum_{i=1}^{n} [ p_i * (\alpha.f + (\alpha.f)^2)  - \alpha.f/n ]\Phi_i
\end{aligned}
\end{equation}
%
%$\Box$\\
\end{proof}
 Simplifying further and observing that $\Phi_i$ decreases and $p_i$ increases with increasing $i$ from $1$ to $n$, one gets the following Corollary.
\begin{corollary}\label{cor-phi-gen}
   $\EE[ \Phi(t+1) - \Phi(t) | x(t) ] \le (\alpha.f)^2 * \Phi / n$ 
\end{corollary}
\begin{proof}
   Since, $p_i$ are increasing and $\Phi_i$ are decreasing, the maximum value taken by RHS of equation~\eqref{eq-phi-base} will be when $p_i = 1/n$ for all $i \in [1..n]$. Simplifying, we get the result.
%
%$\Box$\\
\end{proof}

  Similarly, the change in $\Psi$ can be derived. For detailed proof refer to Appendix~\ref{app:1_proof}.
\begin{lemma}\label{lemma-psi-base}
  When an md-ball is thrown into an md-bin, the following inequality holds:
\begin{equation}\label{eq-psi-base}
   \EE[ \Psi(t+1) - \Psi(t) | x(t) ] \le
   \sum_{i=1}^{n} [ p_i * (-\alpha.f + (\alpha.f)^2)  + \alpha.f/n ]\Psi_i 
\end{equation}
\end{lemma}
Further observing that $p_i > 0$, one gets the following Corollary.
\begin{corollary}\label{cor-psi-gen}
   $\EE[ \Psi(t+1) - \Psi(t) | x(t) ] \le (\alpha.f.\Psi)/ n$ 
\end{corollary}

In the next two lemmas, Lemma~\ref{lemma-phi-easy-case} and Lemma~\ref{lemma-psi-easy-case}, we consider a reasonably balanced md-bins scenario. We show that for such cases, the expected potential decreases. Specifically, for $s_{(n\gamma_2)} \le 0$, the expected value of $\Phi$ decreases and for $s_{(n\gamma_1)} \ge 0$, the expected value of $\Psi$ decreases.
\begin{lemma}\label{lemma-phi-easy-case}
  Let $\Phi$ be defined as above. If $s_{(n\gamma_2)}(t) < 0$ then, $\EE[ \Phi(t+1) | x(t) ] \le (1-\frac{\alpha f}{2n})\Phi$ 
\end{lemma}

\begin{proof}
    From equation~\eqref{eq-phi-base}, we get,
\begin{equation}\label{eq-phi-1}
\begin{aligned}
  \EE[ \Phi(t+1) - \Phi(t) | x(t)] & \le \sum_{i=1}^{n} (p_i * (\alpha f + (\alpha f)^2) - \alpha f/n). \Phi_i\\
   & \le \sum_{i < n\gamma_2} (p_i * (\alpha f + (\alpha f)^2)).\Phi_i - \frac{\alpha f\Phi}{n} + 
         \sum_{i \ge n\gamma_2} p_i * (\alpha f + (\alpha f)^2).e^{-\alpha s_i}
\end{aligned}
\end{equation}
   %The last inequality follows since $\alpha.f < 1/2$ and $\sum_{i \ge (n\gamma_2)} p_i < 1$.
  Now, we need to upper bound the term $\sum_{i < n\gamma_2} (p_i * (\alpha.f + (\alpha.f)^2).\frac{e^{\alpha.s_i}}{i})$. Since $p_i$ is non-decreasing and $\Phi_i$ is non-increasing, the maximum value is achieved when $e^{\alpha s_i} ( \sum_{i=1}^{n\gamma_2} 1/i) = \Phi$ for each $i < n\gamma_2$. Hence, $e^{\alpha s_i} = \frac{\Phi}{\ln(n\gamma_2)}$. Hence, the maximum value is given as follows.
%%%%
\begin{equation}
\begin{aligned}
   \sum_{i=1}^{n\gamma_2} p_i\Phi_i & \le \frac{\Phi}{\ln(n\gamma_2)} * \sum_{i=1}^{n\gamma_2} [ \frac{2i-1}{n^2} * \frac{1}{i} ] \\
       & \le \frac{2\gamma_2\Phi}{n\ln(n\gamma_2)} - \frac{\Phi}{n^2}
\end{aligned}
\end{equation}
%%%%%
  Similarly, one can upper bound the term, $\sum_{i \ge n\gamma_2} (p_i \frac{e^{-\alpha.s_i}}{i})$. Since $p_i$ is non-decreasing and $\Phi_i$ is non-increasing, the maximum value is achieved when $e^{-\alpha s_i} ( \sum_{i=(n\gamma_2)}^{n} 1/i) = \Phi_{(\ge n\gamma_2)}$ for each $i \ge n\gamma_2$. Hence, $e^{-\alpha s_i} = \frac{\Phi_{(\ge n\gamma_2)}}{\ln(1/\gamma_2)}$. Hence, the required upper bound is given as follows.
%%%%
\begin{equation}
\begin{aligned}
   \sum_{i=n\gamma_2}^{n} p_i\Phi_i & \le \frac{\Phi}{\ln(1/\gamma_2)} * \sum_{i=n\gamma_2}^{n} [ \frac{2i-1}{n^2} * \frac{1}{i} ] \\
       & \le \frac{2(1-\gamma_2)\Phi}{n\ln(n)} \because \quad \Phi_{(\ge n\gamma_2)} \le \frac{\Phi}{ln(n)}
\end{aligned}
\end{equation}
%%%%%% 
  Thus, the expected change in $\Phi$ can be computed, using equation~\eqref{eq-phi-1} and the above bound, as follows:
\begin{equation}
\begin{aligned}
  \EE[ \Phi(t+1) - \Phi(t) | x(t)] & \le (\alpha.f + (\alpha.f)^2)(\frac{2\gamma_2\Phi}{n\ln(n\gamma_2)} - \frac{\Phi}{n^2}) - \alpha.f/n * \Phi + (\alpha.f + (\alpha.f)^2)* \frac{2(1-\gamma_2)\Phi}{n\ln(n)}\\ 
   & \le (\alpha.f)\Phi/n(\frac{2\gamma_2}{\ln(n\gamma_2)} - 1) + \frac{\epsilon^2\gamma_2}{2n\ln(n\gamma_2)} + \alpha.f\frac{2(1-\gamma_2)\Phi}{n\ln(n)}\\
   & \le \frac{-\alpha f\Phi}{2n}
\end{aligned}
\end{equation}
%
%
%$\Box$\\
\end{proof}
\begin{lemma}\label{lemma-psi-easy-case}
  Let $\Psi$ be defined as above. If $s_{(n\gamma_1)}(t) \ge 0$ then,
  $\EE[ \Psi(t+1) | x(t) ] \le (1-\frac{\alpha f}{8n})\Psi$ 
\end{lemma}

\begin{proof}
	The proof is similar to that of Lemma~\ref{lemma-phi-easy-case}. See Appendix~\ref{app:2_proof} for details of the proof.
\end{proof}

Now, we consider the remaining cases and show that in case the load across the bins , at time $t$, is not reasonably balanced, then for $s_{n\gamma_2} > 0$, either $\Psi$ dominates $\Phi$ or $\Gamma < c$, where, $c = poly(1/\epsilon)$.
\begin{lemma}
\label{lemma-s-gamma-phi}
   Let, $s_{(n\gamma_2)} > 0$ and $\, \EE[\Delta\Phi|x(t)] \ge -\epsilon\Phi/4n$. Then, either $\Phi < \epsilon\gamma_1 * \Psi$, or $\Gamma < c$ for some $c = poly(1/epsilon)$.
\end{lemma}
\begin{proof}
    From equation~\eqref{eq-phi-base}, we get:
\begin{equation}
\begin{aligned}
   \EE[\Delta\Phi | x(t)] & \le \sum_{i=1}^{n} (p_i * (\alpha.f + (\alpha.f)^2) - \alpha.f/n).\Phi_i\\
    & \le \sum_{i \le n\gamma_3} (p_i * (\alpha.f + (\alpha.f)^2) - \alpha.f/n).\Phi_i + \sum_{i > n\gamma_3} (p_i * (\alpha.f + (\alpha.f)^2) - \alpha.f/n) * \Phi_i\\
    & \le  [(\alpha.f + (\alpha.f)^2)*\sum_{i \le (n\gamma_3)} \frac{2i-1}{n^2i}].\frac{\Phi_{(\le n\gamma_3)}}{\ln(n\gamma_3)} +\\ 
    & ((\sum_{i > (n\gamma_3)} \frac{2i-1}{n^2i}))* (\alpha.f + (\alpha.f)^2)\frac{\Phi_{(> n\gamma_3)}}{\ln(1/\gamma_3)} - \frac{\alpha f\Phi}{n}\\
    & \le \frac{\alpha f\Phi_{\le n\gamma_3}}{n} [\frac{2\gamma_3}{\ln(n\gamma_3)} - 1 - 1/n] + \frac{\alpha f\Phi_{> n\gamma_3}}{n} [\frac{2\gamma_4}{\ln(1/\gamma_3)} - 1 - 1/n]\\
    & \le \frac{\alpha f\Phi_{\le n\gamma_3}}{n} [\frac{2\gamma_3 - \ln(n\gamma_3)}{\ln(n\gamma_3)}] + \frac{\alpha f\Phi_{>n\gamma_3}}{n} [\frac{2\gamma_4 - \ln(1/\gamma_3)}{\ln(1/\gamma_3)}]
\end{aligned}
\end{equation}
  Now, since $\EE[\Delta\Phi|x(t)] \ge -\alpha f\Phi/2n$, we get: $\Phi \le 4\Phi_{(> n\gamma_3)} [\frac{\gamma_4\ln(n)}{\ln(n\gamma_3)\ln(1/\gamma_3)}]$.
  Let, $B = \sum_{i} max(0, s_i)$. Note, $\sum_{i} s_i = 0$, since for each dimension $d$, the update maintains that, $\sum_{d=1}^{D} x_i^d(t) = 0$. Further, because, $s_{n\gamma_3} > 0$, $\Phi_{(>n\gamma_3)} \le \ln(1/\gamma_3) * e^{(\frac{\alpha.B}{n\gamma_3})}$. This implies that, $\Phi \le \frac{4\gamma_4\ln(n) e^{(\frac{\alpha.B}{n\gamma_3})}}{\ln(n\gamma_3)}$.

Since, $s_{(n\gamma_2)} > 0$, so, $\Psi \ge \ln(n\gamma_2) * e^{\frac{\alpha.B}{n\gamma_1}}$. If $\Phi < \epsilon\gamma_1 * \Psi$, then we are done. Else, $\Phi \ge \epsilon\gamma_1 * \Psi$. This implies:
\begin{equation}\nonumber
  \frac{4\gamma_4\ln(n)e^{\frac{\alpha.B}{n\gamma_3}}}{\ln(n\gamma_3)} \ge \Phi \ge \epsilon\gamma_1 * \Psi \ge \epsilon\gamma_1\ln(n\gamma_2) * e^{\frac{\alpha.B}{n\gamma_1}}
\end{equation}
  Thus, $e^{\alpha.B/n} \le (\frac{4\gamma_4}{\epsilon\gamma_1})^{\frac{\gamma_1\gamma_3}{(\gamma_3 - \gamma_1)}}$. So, $\Gamma \le (\frac{1 + \theta}{\epsilon} * \Phi \le \frac{1 + \theta}{\epsilon} * \frac{4\gamma_4\ln(n)}{\ln(n\gamma_3)} * e^{\frac{\alpha.B}{n\gamma_3}} \le \frac{1 + \theta}{\epsilon} * \frac{4\gamma_4\ln(n)}{\ln(n\gamma_3)} * (\frac{4\gamma_4}{\epsilon\gamma_1})^{\frac{\gamma_1}{(\gamma_3 - \gamma_1)}}$. Hence, $\Gamma < c$, where, $c = poly(1/\epsilon)$.
%
%$\Box$\\
\end{proof}

In the Lemma below, we consider the case where the load across the bins at time, $t$, is not reasonably balanced, and $s_{(n\gamma_1)} < 0$. Here, we show that either $\Phi$ dominates $\Psi$ or the potential function is less than $c$ for $c = poly(1/\epsilon)$.
\begin{lemma} \label{lemma-s-gamma-psi}
   Let, $s_{(n\gamma_1)} < 0$ and $\EE[\Delta\Psi|x(t)] \ge -\alpha f\Psi/8n$. Then, either $\Psi < \epsilon\gamma_1 * \Phi$, or $\Gamma < c\ln(n)$ for some $c = poly(1/epsilon)$.
\end{lemma}

\begin{proof}
	The proof is similar to that of Lemma~\ref{lemma-s-gamma-phi}. See Appendix~\ref{app:3_proof} for details of the proof.
\end{proof}
  Now, we consider combinations of the cases considered so far and can show that the potential function, $\Gamma$, behaves as a super-martingale.
\begin{theorem}\label{thm-super-mart}
   For the potential function, $\Gamma$, $\EE[ \Gamma(t+1) | x(t)] \le (1 - \frac{\epsilon}{24n(1 + \epsilon\gamma_1)})\Gamma(t) + \frac{c\ln(n)}{n}$, for constant $c = poly(1/\epsilon)$.
\end{theorem}
\begin{proof}
    We consider the following cases on intervals of values for $s_i$.
\begin{itemize}
\item \textbf{Case 1:} $s_{(n\gamma_1)} \ge 0$ and $s_{(n\gamma_2)} \le 0$. Using, Lemma~\ref{lemma-phi-easy-case} and Lemma~\ref{lemma-psi-easy-case}, we can immediately see that, $\EE[ \Gamma(t+1) | x(t)] \le (1 - \epsilon^2/16n)\Gamma(t)$ and hence, the result is also true.\\
\item \textbf{Case 2:} $s_{n\gamma_1} \ge s_{n\gamma_2} > 0$. This represents a high load imbalance across the bins. In some cases, $\Phi$ may grow but the asymmetry in the load implies that $\Gamma$ is dominated by $\Psi$. Thus, the decrease in $\Psi$ offsets the increase in $\Phi$ and hence the expected change in $\Gamma$ is negative.\\
  Specifically, if $\EE[\Delta\Phi|x] \le \frac{-\alpha f\Phi}{2n}$, then using Lemma~\ref{lemma-psi-easy-case} we get that $\EE[ \Gamma(t+1) | x(t)] \le (1 - \alpha f/8n)\Gamma(t)$; else we consider the following two cases:\\
\begin{itemize}
\item \textbf{Case 2a:} $\Phi < \epsilon\gamma_1 * \Psi$. Here, using Lemma~\ref{lemma-psi-easy-case} and Corollary~\ref{cor-phi-gen}, we get:
\begin{equation}\nonumber
\begin{aligned}
   \EE[\Delta\Gamma|x] & = \EE[\Delta\Phi|x] + \EE[\Delta\Psi|x]\\
        & \le \frac{(\alpha.f)^2}{n}.\Phi - \frac{\alpha f}{8n} * \Psi\\
        & \le -\frac{\epsilon}{24n}.\Psi\\
        & \le - \frac{\epsilon}{24n(1 + \epsilon\gamma_1)}\Gamma
\end{aligned}
\end{equation}
\item \textbf{Case 2b:} $\Gamma < c\ln(n)$. Here, using Corollary~\ref{cor-psi-gen} and Corollary~\ref{cor-phi-gen}, we get:
\begin{equation}\nonumber
\begin{aligned}
   \EE[\Delta\Gamma|x] & \le \alpha.f/n * \Gamma
        & \le \frac{c\alpha.f\ln(n)}{n}
\end{aligned}
\end{equation}
   But, $c\ln(n)/n - ((\epsilon/4n) * \Gamma) \ge c\ln(n)/n(1 - \epsilon/4) \ge c\ln(n)/n(1 - \epsilon/2) \ge \frac{c\alpha.f\ln(n)}{n}$. Hence,$\EE[\Delta\Gamma|x] \le - \frac{\epsilon\Gamma}{4n} + \frac{c\ln(n)}{n}$.
\end{itemize}
\item \textbf{Case 3:} $s_{n\gamma_2} \le s_{n\gamma_1} < 0$. Here, if $\EE[\Delta\Psi|x] \le \frac{-\epsilon}{16n}\Psi$, then using Lemma~\ref{lemma-phi-easy-case}, we get that $\EE[ \Gamma(t+1) | x(t)] \le (1 - \epsilon/16n)\Gamma(t)$; else we consider the following two cases:
\begin{itemize}
\item \textbf{Case 3a:} $\Psi < \epsilon\gamma_1 * \Phi$. Here, using Lemma~\ref{lemma-phi-easy-case} and Corollary~\ref{cor-psi-gen}, we get:
\begin{equation}\nonumber
\begin{aligned}
   \EE[\Delta\Gamma|x] & = \EE[\Delta\Phi|x] + \EE[\Delta\Psi|x]\\
        & \le -(\epsilon/4n)\Phi + \alpha.f/n * \Psi\\
        & \le -(\epsilon/4n).\Phi + (\gamma_1\epsilon^2/2n)\Phi\\
        & \le \frac{-\epsilon}{8n}\Phi\\
        & \le \frac{-\epsilon}{(8n(1+\epsilon\gamma_1))}*\Gamma
\end{aligned}
\end{equation}
\item \textbf{Case 3b:} $\Gamma < c\ln(n)$. Here, using Corollary~\ref{cor-phi-gen} and Corollary~\ref{cor-psi-gen}, we get:
\begin{equation}\nonumber
\begin{aligned}
   \EE[\Delta\Gamma|x] & = \alpha.f/n * \Gamma
        & \le c\alpha.f\ln(n)/n
\end{aligned}
\end{equation}
  Hence, this case follows similarly as \textit{Case 2b} above. 
\end{itemize}
\end{itemize}
%
%$\Box$\\
\end{proof}
  Now, we can prove using induction that the expected value of $\Gamma$ remains bounded.
\begin{theorem}\label{thm-pot-up-bound}
  For any time $t \ge 0$, $\EE[\Gamma(t)] \le \frac{24c(1 + \epsilon\gamma_1)}{\epsilon}\ln(n)$
\end{theorem}
\begin{proof}
  Using induction we can prove this claim. For $t = 0$, it is trivially true since $\Gamma(0) = 2\ln(n)$.
Using Theorem~\ref{thm-super-mart}, we get:
\begin{equation}\nonumber
\begin{aligned}
   \EE[\Gamma(t+1)] & = E[E[\Gamma(t+1)| \Gamma(t)]]\\
            & \le \EE[ (1 - \frac{\epsilon}{24n(1+\epsilon\gamma_1)})\Gamma(t) + \frac{c\ln(n)}{n}]\\
            & \le \frac{24c(1 + \epsilon\gamma_1)}{\epsilon}\ln(n) - c\frac{\ln(n)}{n} + \frac{c\ln(n)}{n}\\
            & \le \frac{24c(1 + \epsilon\gamma_1)}{\epsilon}\ln(n)
\end{aligned}
\end{equation}
%
%$\Box$\\
\end{proof}
%%%%%
%
 Now, we can upper bound the gap across all the $D$ dimensions across all the $n$ md-bins. This gap is defined as follows:
\begin{equation}
    Gap(t) = \max_{d=1}^{D} [ \max_{i=1}^{n} x_i^d ]
\end{equation}
%
%%%%%%%%%%%%%%%%
%%%%%%%% For Fixed "`f" Case, m>>n %%%%%%%
%%%%%%%%%%% sub-cases: uniform(f/D) and non-uniform(prob per dim >= c)
%%%%%%%%%%%%%%%%
%%%%%%%%%%%%%
\begin{theorem}\textbf{Fixed $f$ Case:}\label{thm-fixed-f}
  Using the bias in the probability distribution in favor of lightly loaded md-bins as given by the $d$-choice algorithm, and assuming that $f$ dimensions are exactly populated in each md-ball with uniform distribution of $f$ dimensions over $D$, the expected and probabilistic upper bound on the gap (maximum dimensional gap) across the multidimensional bins is given as follows. Let, $\delta = \frac{24c(1 + \epsilon\gamma_1)}{\epsilon}$, then:
\begin{equation}\nonumber
\begin{aligned}
    E[Gap(t)] & \le 2\log\log(n)/\epsilon + 2\log\log(\delta)/\epsilon\\
    Pr[Gap(t) & > (\frac{mf}{nD})^{1/2+\zeta}*(4\log\log(n)/\epsilon + 4\log(\log\delta)/\epsilon)] \le D/mf\\
\end{aligned}
\end{equation}
\end{theorem}
\begin{proof}

Let, $a$ be the winning md-bin and $m$ be the winning dimension that represents $Gap(t)$. Now,
from Theorem~\ref{thm-pot-up-bound}, we get, $\EE[ e^{\alpha.s_a} ] \le \delta\log(n)$. So, 
$\EE[ e^{\alpha.(x_a^m + \sum_{d \ne m} x_a^d)} ] \le \delta\log(n)$. 
  Let, $y_a$ denote the gap as measured by the number of md-balls in bin $a$ minus the average number of balls across the bins. Then,
\begin{equation}\label{eq:sum-bound}
\begin{aligned}
    E[s_a] & \le 1/\alpha * \log\log(n) + 1/\alpha * \log\log(\delta)\\
    \Rightarrow E[s_a] & \le 2f\log\log(n)/\epsilon + O(2f\log\log(\delta)/\epsilon)\\
    \Rightarrow f.E[y_a] & \le 2f\log\log(n)/\epsilon + O(2f\log\log(\delta)/\epsilon)\\
    \Rightarrow E[y_a] & \le 2\log\log(n)/\epsilon + O(2\log\log(\delta)/\epsilon)
\end{aligned}
\end{equation}
  The third inequality uses the fact that each ball has exactly $f$ populated dimensions. Since the $f$ dimensions are chosen uniformly and randomly from $D$ dimensions, the expected gap in any dimension (and hence the winning dimension with the maximum gap) is bounded by $O(\log\log(n))$.
  Now, consider the case of a non-uniform distribution, where we assume that each dimension is chosen with probability at most $\kappa_2$ in each md-ball and each md-ball still has fixed $f$ populated dimensions. Here, one can see that the expected gap can be bounded by $O(\kappa_2\log\log(n))$.

  %Now, the maximum value of RHS is attained when $|x_a^d|$ is the maximum. Note, that $|x_a^d| \le x_a^m$, since for each dimension $d$, $|x_a^d| < x_b^d$, where $b$ is the md-bin that has the maximum value for dimension $d$, over all the bins. By definition of $x_a^m$, it represents maximum gap across all bins and all dimensions, so $x_a^d \le x_a^m$. Hence, we get:
%
Now, the $Pr[s_a > 4f\log\log(n)/\epsilon + 4f\log\log(\delta)/\epsilon] \le
Pr[ \Gamma(t) \ge nE[\Gamma(t)]] \le 1/n$ (using Markov's Inequality); where $s_a = \sum_{d=1}^{D} x_a^d$. Further, the probability that within a single md-bin, a particular dimension has more than the expected number of $1s$, can be given by the Chernoff Bound as follows. Let $m/n$ balls be thrown into an md-bin. The number of ones in any dimension follows a Binomial distribution, $B(m/n, f/D)$. Using Chernoff Bound, and assuming $t = (\frac{mf}{nD})^{1/2 + \zeta}$, we have:
\begin{equation}
\begin{aligned}
   Pr[B(m/n, f/D) > (mf/nD + t)] & \le (\frac{mf/nD}{mf/nD + t})^{mf/nD + t} * e^{t}\\
   \Rightarrow Pr[B(m/n, f/D) > (mf/nD + t)] & \le nD/mf
\end{aligned}
\end{equation}

Hence, $Pr[y_a > (\frac{mf}{nD})^{1/2 + \zeta}*(4\log\log(n)/\epsilon + 4\log\log(\delta)/\epsilon)] \le 1/n * nD/mf = D/mf$. 
%
%$\Box$
\end{proof}
%%%%
%%%%
%%%%%%
%%%%%%
%%%%%%
%%%%
%OLD => One can see that for, $m \le n\log(n)$, the gap becomes $O(\log(n)/\beta)$. Below, we give a matching lower bound, hence demonstrating that the upper bound is tight. Further, for even non-uniform distributions, when the probability of choosing a subset of $f$ populated dimensions is lower bounded by a constant $c$, the expected average load per dimension per bin is bounded by $mc/n$. Hence, $\Rightarrow E[x_a^m] \le 2\log(n)/\epsilon + 2\log(\delta)/\epsilon + mc/n(D/f - 1)$.
%%%%%%
%%%%%%%%%%% END OF LOOSE ANALYSIS %%%%%%%%%%%%%%%%%%%%%%%%%%
%%%%%%
%%%%%%
\subsection{Lower Bound for Unweighted Case}

We can show that the expected upper bound, for fixed $f$ case with uniform distribution, proved in section~\ref{sec:unweighted} is tight to within $f/D$ factor. Consider the case when, $m$ balls are thrown into $n$ bins, using the $d$ choice process. The expected dimensional sum load per bin is $fm/n$. Berenbrink et.al.~\cite{petra-heavy-case} show that when $m >> n$ balls are thrown using the $d$ choice process into $n$ bins, then the load of the most loaded bin is at least $O(\ln\ln(n))$ balls more than the average $m/n$. Thus, for md-balls the sum load of the most loaded md-bin is at least $\Omega(f\ln\ln(n) + fm/n)$. Since, each ball has $f$ populated dimensions, hence, there are at least $\Omega(\ln\ln(n) + m/n)$ balls in this max sum load bin. Since, in each ball $f$ dimensions are uniformly distributed over $D$ dimensions, there exists a dimension whose load is at least $\Omega(f\ln\ln(n)/D)$ more than the average $mf/nD$. Hence, the lower bound is $O(f\ln\ln(n)/D)$.
%%%%%%%
%%%%%%%
%%%%%%%
%%%%%%%
\subsection{Parallel Multidimensional Balls \& Bins: Unweighted Case}
%%%%
\par Consider the following parallel $d$-choice process. Let $m$ balls be thrown in parallel using $d$-choice process into $n$ bins. In each round, a bin sends the (ball's) rank to the ball with the lowest ID. The ball chooses the bin (out of $d$ bins it selected) that gives the lowest rank.  It can be shown that this parallel process produces exactly the same distribution of balls in the bins as a sequential \textit{Greedy with Ties} process~\cite{adler95}. In the sequential \textit{Greedy with Ties} process, when there are multiple bins with same lowest load, all of these bins get the ball. Using the potential function analysis as above, we can show that the gap in this case, can also be bounded by $O(\log\log(n))$. We provide an overview of the proof below.
%%%%%
Consider the change in $\Phi(t)$ (also refers to $\Phi$ by default) and $\Psi(t)$ (also refers to $\Psi$ by default) separately when a ball is thrown with the given probability distribution. 
\begin{lemma}\label{lemma-phi-base-par}
  When an md-ball is thrown into an md-bin, the following inequality holds:
\begin{equation}
   \EE[ \Phi(t+1) - \Phi(t) | x(t) ] \le \sum_{i=1}^{n} [ d*p_i * (\alpha.f + (\alpha.f)^2)  - d\alpha.f/n ]. e^{\alpha.s_i} 
\end{equation}
\end{lemma}
\begin{proof}
  In the Greedy with Ties process, some number (less than $d$) of bins each with the same load (and hence belonging to the same equi-load group) can get the (\textit{replicated}) ball. In the worst case all the $d$ randomly selected bins, chosen by the ball, have the same load and hence get the md-ball. All of these md-bins, then move to the previous equi-load group or a new equi-load group is created. Let $\Delta$ be the expected change in $\Phi$ when the ball is put in a certain number (less than $d$) of bins. If one of these bins is $i$, then, $r_i(t+1) = s_i + f(1-d/n)$. For bins $j \ne i$, that do not get the md-ball, $r_j(t+1) = s_j(t) - df/n$. The new values i.e. $s(t+1)$ are obtained by sorting $r(t+1)$ and $\Phi(s) = \Phi(r)$.  Thus, the expected contribution of bin, $i$, to $\Delta$ is given as follows:
\begin{equation}\nonumber
\begin{aligned}
   \EE[ \frac{e^{\alpha.(s_i + f(1-d/n))}}{i} ] - \frac{e^{\alpha.s_i}}{i}
   & = \frac{e^{\alpha.s_i}}{i} [ e^{\alpha.f(1-d/n)} - 1 ]
\end{aligned}
\end{equation}
 Similarly, the expected contribution of bin (that does not get the  ball), $j$ ($j \ne i$) to $\Delta$ is given as:
\begin{equation}\nonumber
\begin{aligned}
   \EE[ \frac{e^{\alpha.(s_j - df/n)}}{j} ] - \frac{e^{\alpha.s_j}}{j}
   & = \frac{e^{\alpha.s_j}}{j} [ e^{-\alpha.df/n} - 1 ]
\end{aligned}
\end{equation}
 Assuming that the bins that get the replicated ball are $i_1, i_2,..i_d$, $\Delta$ is given as follows:
\begin{equation}\nonumber
\begin{aligned}
   \Delta & = d\Phi_i [ e^{\alpha.f(1-d/n)} - 1] \, + \,
              \sum_{j \ne (i_1, i_2,..i_d)} \Phi_j(e^{-\alpha.df/n} - 1)\\
   & = d\Phi_ie^{-\alpha df/n} (e^{\alpha.f} - 1) + (e^{-\alpha.df/n} - 1).\Phi
\end{aligned}
\end{equation}
  Thus, we get the overall expected change in $\Phi$ as follows:
\begin{equation}
\begin{aligned}
   \EE[ \Phi(t+1) - \Phi(t) | x(t) ] & = \sum_{i=1}^{n} p_i * \Delta\\
   & = \sum_{i=1}^{n} p_i * [ d\Phi_ie^{-\alpha df/n} (e^{\alpha.f} - 1) + (e^{-\alpha.fd/n} - 1).\Phi ]\\
   & = \sum_{i=1}^{n} p_i * de^{-\alpha df/n}\Phi_i (e^{\alpha.f} - 1) +                        (e^{-\alpha.df/n} - 1).p_i.\Phi\\
   & = \sum_{i=1}^{n} [ p_i * de^{-\alpha.df/n} (e^{\alpha.f} - 1) + (e^{-\alpha.df/n} - 1) ]\Phi_i
\end{aligned}
\end{equation}
 Now, $e^{(-\alpha.fd/n)} * (e^{\alpha.f} - 1)$ can be approximated as follows:
\begin{equation}\nonumber
\begin{aligned}
   e^{(-\alpha.df/n)}.(e^{\alpha.f} - 1) & \le
   (1 - \alpha.df/n + (\alpha.df/n)^2) * (1 + \alpha.f + (\alpha.f)^2 - 1)\\
   & \sim \alpha.f + (\alpha.f)^2 + O((\alpha.df)^2/n)\\   
   \Rightarrow e^{(-\alpha.df/n)}.(e^{\alpha.f} - 1) & \lessapprox (\alpha.f + (\alpha.f)^2)
\end{aligned}
\end{equation}
Above, since, $(\alpha.fd)^2/n$ is very small for large $n$, we have ignored the small terms. Similarly, $(e^{-\alpha.df/n} - 1) \lessapprox -\alpha.df/n$
Hence, the expected change in $\Phi$ can be given by:
\begin{equation}\label{eq-phi-base}
\begin{aligned}
   \EE[ \Phi(t+1) - \Phi(t) | x(t) ]
   \le \sum_{i=1}^{n} [ p_i * d(\alpha.f + (\alpha.f)^2)  - \alpha.df/n ]\Phi_i
\end{aligned}
\end{equation}
%
%$\Box$\\
\end{proof}
%%%%%
Similarly, one can show the following.
\begin{lemma}\label{lemma-psi-base-par}
  When an md-ball is thrown into an md-bin, the following inequality holds:
\begin{equation}
   \EE[ \Psi(t+1) - \Psi(t) | x(t) ] \le \sum_{i=1}^{n} [ d*p_i * (-\alpha.f + (\alpha.f)^2)  + d\alpha.f/n ]. e^{-\alpha.s_i} 
\end{equation}
\end{lemma}
%%%%%
%%%%%
%%%%%
Following similar lines of proof as for the sequential multidimensional case, one can hence show that:\\
\begin{theorem}\label{thm-pot-up-bound-par}
  For any time $t \ge 0$, $\EE[\Gamma(t)] \le \frac{24cd(1 + \epsilon\gamma_1)}{\epsilon}\log(n)$
\end{theorem}
%%%%%
Thus, this parallel balls and bins process with $m >> n$ balls and $n$ bins, takes $O(\frac{m}{n} + \log\log(n))$ rounds and results in maximum bin load $O(\frac{m}{n} + \log\log(n))$ resulting in upper bound on the gap of $O(\log\log(n))$. Hence, one can derive that the gap (using Theorem~\ref{thm-fixed-f}) for the multidimensional parallel scenario is also bounded by $O((mf/D)^{(1/2+\zeta)}\log\log(n))$ with high probability.
%%%%%
%%%%%
%%%%%%%%%%% Weighted Case: Detailed Analysis %%%%%%%%%%%%%%%
%%%%%
%%%%%
%%%%%%
%%%%%%
\subsection{Upper Bound On Gap: Weighted Case}\label{sec:weighted}
%%%%%
\par Here, we consider the case when the multidimensional balls have variable number of populated dimensions, $f$. The sum of dimensional load in an md-ball, $f$, is thus a random variable. We assume that the distribution for $f$ has a finite second moment and average value, $f^*$. For this distribution, we assume that there is a
$\lambda > 0$ such that the moment generating function $M[\lambda] =
E[e^{\lambda.f} ] < \infty$. Note that $M''(z) = E[f^2e^{zf}] \le \sqrt{E[f^4]E[e^{2zf}]}$. The above assumption implies that there is a $S \ge 1$,
such that for every $|z| < \lambda/2$ it holds that $M''(z) < 2S$. Our analysis below is primarily for integer valued $f$ and for the multidimensional case. However, it can be easily seen that similar analysis holds for scalar balls and bins with real valued weight per ball $W$ and $E[W] = 1$ (still assuming that the distribution of $W$ has finite second moment). 
%%%%
\par The weighted case is more challenging that the unweighted case, since we have to carefully consider the change in the rank of a bin when an md-ball of total dimensional load (weight) $f$ falls in it, as the change in rank could increase the potential by a large amount. Thus, the potential function used in section~\ref{sec:unweighted} might not work in this case and we need to devise a new one. Assume that $\epsilon \le 1/4$. Further, let $\alpha = \min{(\frac{\epsilon}{6S}, \frac{2}{\lambda}, \frac{\epsilon}{2f^*})}$. Define the following potential functions over the bins:
\begin{equation}
\begin{aligned}
     \Phi(t) & = \Phi(s(t)) = \sum_{i=1}^{n} \frac{e^{\alpha.s_i}}{n^2+i}\\
     \Psi(t) & = \Psi(s(t)) = \sum_{i=1}^{n} \frac{e^{-\alpha.s_i}}{n^2+n-i+1}\\     
     \Gamma(t) & = \Gamma(s(t)) = \Phi(t) + \Psi(t)\\
               & =  \sum_{i=1}^{n} [ \frac{e^{\alpha.s_i}}{n^2+i} + \frac{e^{-\alpha.s_i}}{n^2+n-i+1} ]\\   
\end{aligned}
\end{equation}
  where, $s_i(t) = \sum_{d=1}^{D} x_i^d(t)$

In the beginning, each dimension for each bin has $0$ weight, thus $s_i = 0, \forall  i$ and hence, $\Gamma(0) \le 2(n/(n^2 + 1)) \le 2/n$. We show that if $\Gamma(x(t)) \ge a/n$ for some $a > 0$, then $\mathbb{E}[\Gamma(t+1) | x(t)] \le (1 - \frac{\alpha.f^*}{16n(1 + \epsilon\gamma_1)}) * \Gamma(t)$. This helps in demonstrating that for every given $t$, $\mathbb{E}[\Gamma(t)] \in O(1/n)$. This implies that the maximum gap is $O(\log(n))$ w.h.p.

First, consider the change in $\Phi(t)$ (also refers to $\Phi$ by default) and $\Psi(t)$ (also refers to $\Psi$ by default) separately when a ball is thrown with the given probability distribution. Let there be constants, $0 < \gamma_1 < \gamma_2 < 1/2 < \gamma_4 < \gamma_3$, such that $\gamma_2 + \gamma_3 > 1$ and $\gamma_1 + \gamma_4 < 1$ and $\gamma_2 < 7/16$
%%%%%%%%%%
%%%%%%%%%%
%%%%%%% Weighted Case: Phi General Case %%%%%%%%%
%%%%%%%%%%
%%%%%%%%%%
\begin{lemma}\label{lemma-phi-base-w}
  When an md-ball is thrown into an md-bin, the following inequality holds:
\begin{equation}
   \EE[ \Phi(t+1) - \Phi(t) | x(t) ] \le \sum_{i=1}^{n} [ p_i * (\alpha.f^* + 1/n + S\alpha^2)  - \alpha.f^*/n ]. e^{\alpha.s_i} 
\end{equation}
\end{lemma}
\begin{proof}
  Let $\Delta_i$ be the expected change in $\Phi$ if the ball is put in bin, $i$. So, $r_i(t+1) = s_i + f(1-1/n)$; and for $j \ne i$, $r_j(t+1) = s_j(t) - f/n$. The new values i.e. $s(t+1)$ are obtained by sorting $r(t+1)$ and $\Phi(s) = \Phi(r)$. When, an md-ball is committed to bin $i$, then it jumps to an index $i_{new}$ which is less than or equal to $i$ in the new bin order. Thus, the expected contribution of bin, $i$, to $\Delta_i$ is given as follows:
\begin{equation}\label{eq-jump}
\begin{aligned}
   \EE[ \frac{e^{\alpha.(s_i + f(1-1/n))}}{n^2 + i_{new}} ] - \frac{e^{\alpha.s_i}}{n^2 + i}\\
   & \le \frac{e^{\alpha.s_i}}{n^2 + i} [ \frac{M(\alpha(1-1/n))(n^2+n)}{n^2 + 1} - 1 ]\\
   & \le \Phi_i [ (M(0) + M'(0).\alpha(1-1/n) + M''(0)(\alpha(1-1/n))^2)(1+1/n) - 1  ]\\
   & \le \Phi_i [ (1 + f^*\alpha(1-1/n) + S\alpha^2)(1 + 1/n) - 1] \\
            &\because \quad M(0) = 1, M'(0) = E(f) = f^*, M''(0) \le 2S\\
   & \le \Phi_i [ f^*\alpha(1-1/n) + S\alpha^2 + (1 + f^*\alpha(1-1/n) + S\alpha^2)/n]\\
   & \le \Phi_i [ f^*\alpha + 1/n + S\alpha^2]
\end{aligned}
\end{equation}
 The bins that were at index $j \in [i_{new}..(i-1)]$, shift right by one position and hence the expected contribution of such a bin, $j$ to $\Delta_i$ is given as:
\begin{equation}\label{eq-one-shift}
\begin{aligned}
   & \EE[ \frac{e^{\alpha.(s_j - f/n)}}{n^2 + j+1} ] - \frac{e^{\alpha.s_j}}{n^2 + j}\\
   & \le \frac{e^{\alpha.s_j}}{n^2 + j} [ M(-\alpha/n) - 1 ]\\
   & \le \Phi_j [ M(0) + M'(0)(-\alpha/n) + M''(0)\frac{\alpha^2}{2n^2}]\\
   & \le \Phi_j [ \frac{-f^*\alpha}{n} + \frac{S\alpha^2}{n^2} ]\\
\end{aligned}
\end{equation}
%%%%%
  For all other bins, their rank does not change in the new bin order, hence, their expected contribution to $\Delta_i$ is given as:
\begin{equation}\label{eq-no-shift}
\begin{aligned}
   & \EE[ \frac{e^{\alpha.(s_j - f/n)}}{n^2 + j} ] - \frac{e^{\alpha.s_j}}{n^2 + j}\\
   & = \Phi_j [ \frac{-f^*\alpha}{n} + \frac{S\alpha^2}{n^2} ]
\end{aligned}
\end{equation}

 Using equations~\eqref{eq-jump},~\eqref{eq-one-shift} and~\eqref{eq-no-shift}, $\Delta_i$ is given as follows:
\begin{equation}\nonumber
   \Delta_i = (-\alpha.f^*/n + 1/n + S\alpha^2)\Phi_i + \frac{\alpha.f^*\Phi}{n}
\end{equation}
%
%
%  Thus, we get the overall expected change in $\Phi$ as follows:
%\begin{equation}
%\begin{aligned}
%   \EE[ \Phi(t+1) - \Phi(t) | x(t) ] & = \sum_{i=1}^{n} p_i * \Delta_i\\
%   & = \sum_{i=1}^{n} p_i * [ \Phi_ie^{-\alpha f/n} (e^{\alpha.f} - 1) + %(e^{-\alpha.f/n} - 1).\Phi ]\\
%   & = \sum_{i=1}^{n} p_i * e^{-\alpha f/n}\Phi_i (e^{\alpha.f} - 1) +               %         (e^{-\alpha.f/n} - 1).p_i.\Phi\\
%   & = \sum_{i=1}^{n} [ p_i * e^{-\alpha.f/n} (e^{\alpha.f} - 1) + (e^{-\alpha.f/n} - 1) ]\Phi_i
%\end{aligned}
%\end{equation}
%
% Now, $e^{(-\alpha.f/n)} * (e^{\alpha.f} - 1)$ can be approximated as follows:
%\begin{equation}\nonumber
%\begin{aligned}
%   e^{(-\alpha.f/n)}.(e^{\alpha.f} - 1) & \le
%   (1 - \alpha.f/n + (\alpha.f/n)^2) * (1 + \alpha.f + (\alpha.f)^2 - 1)\\
%   & \sim \alpha.f + (\alpha.f)^2 + O((\alpha.f)^2/n)\\   
%   e^{(-\alpha.f/n)}.(e^{\alpha.f} - 1) & \lessapprox (\alpha.f + (\alpha.f)^2)
%\end{aligned}
%\end{equation}
%
%Above, since, $(\alpha.f)^2/n$ is very small for large $n$, we have ignored the small terms. Similarly, $(e^{-\alpha.f/n} - 1) \lessapprox -\alpha.f/n$
%
Hence, the expected change in $\Phi$ can be given by:
\begin{equation}\label{eq-phi-base-w}
\begin{aligned}
   \EE[ \Phi(t+1) - \Phi(t) | x(t) ]
   \le \sum_{i=1}^{n} [ p_i * (\alpha.f^* + 1/n + S\alpha^2)  - \alpha.f^*/n ]\Phi_i
\end{aligned}
\end{equation}
%
%$\Box$\\
\end{proof}
 Simplifying further and observing that $\Phi_i$ decreases and $p_i$ increases with increasing $i$ from $1$ to $n$, one gets the following Corollary.
\begin{corollary}\label{cor-phi-gen-w}
   $\EE[ \Phi(t+1) - \Phi(t) | x(t) ] \le (\alpha.f^* + 2S\alpha^2)\frac{\Phi}{n}$ 
\end{corollary}
\begin{proof}
   Since, $p_i$ are increasing and $\Phi_i$ are decreasing, the maximum value taken by RHS of equation~\eqref{eq-phi-base-w} will be when $e^{\alpha.s_i} * \sum_{i=1}^{n} \frac{1}{(n^2 + i)} = \Phi$. Thus, $e^{\alpha.s_i} = n\Phi$. Hence,
\begin{equation}
\begin{aligned}
  \sum_{i=1}^{n} p_i\Phi_i &\le n\Phi * \sum_{i=1}^{n} \frac{2i - 1}{n^2(n^2 + i)}\\
  &\le \frac{\Phi}{n} * \frac{2n-1}{n+1}
\end{aligned}
\end{equation} 
%%%%%%
  Using, equation~\eqref{eq-phi-base-w}, we get:
\begin{equation}
\begin{aligned}
  \EE[ \Phi(t+1) - \Phi(t) | x(t) ] &\le (\alpha.f^*(1+1/n) + 1/n + S\alpha^2) * \frac{(2n-1)\Phi}{n(n+1)} - \frac{\alpha.f^*\Phi}{n}\\
     &\le (\alpha.f^* + 2S\alpha^2)\frac{\Phi}{n}
\end{aligned}
\end{equation}
%%%%%%
%%%%%%
%
%$\Box$\\
\end{proof}
%%%%%%%%%%%%%
%%%%%%%%%%%%%
\par Similarly, the change in $\Psi$ can be derived. For detailed proof refer to Appendix~\ref{app:weighted_1_proof}.
\begin{lemma}\label{lemma-psi-base-w}
  When an md-ball is thrown into an md-bin, the following inequality holds:
\begin{equation}\label{eq-psi-base-w}
   \EE[ \Psi(t+1) - \Psi(t) | x(t) ] \le
   \sum_{i=1}^{n} [ p_i * (-\alpha.f^* + \frac{S\alpha^2)}{n^2}  + \alpha.f^*/n ]\Psi_i 
\end{equation}
\end{lemma}
Further observing that $p_i > 0$, one gets the following Corollary.
\begin{corollary}\label{cor-psi-gen-w}
   $\EE[ \Psi(t+1) - \Psi(t) | x(t) ] \le (\alpha.f^*\Psi)/ n$ 
\end{corollary}

In the next two lemmas, Lemma~\ref{lemma-phi-easy-case-w} and Lemma~\ref{lemma-psi-easy-case-w}, we consider a reasonably balanced md-bins scenario. We show that for such cases, the expected potential decreases. Specifically, for $s_{(n\gamma_2)} \le 0$, the expected value of $\Phi$ decreases and for $s_{(n\gamma_1)} \ge 0$, the expected value of $\Psi$ decreases.
%%%%%
%%%%%
%%%%%%%%%%%%% Weighted Case: \Phi Easy Case %%%%%%%%%%%%
%%%%%
%%%%%
\begin{lemma}\label{lemma-phi-easy-case-w}
  Let $\Phi$ be defined as above. If $s_{(n\gamma_2)}(t) < 0$ then, $\EE[ \Phi(t+1) | x(t) ] \le (1-\frac{\alpha f^*}{8n})\Phi$ 
\end{lemma}

\begin{proof}
    From equation~\eqref{eq-phi-base-w}, we get,
\begin{equation}\label{eq-phi-w-1}
\begin{aligned}
  \EE[ \Phi(t+1) - \Phi(t) | x(t)] & \le \sum_{i=1}^{n} (p_i * (\alpha f^* + S(\alpha f)^2) - \alpha f^*/n). \Phi_i\\
   & \le \sum_{i < n\gamma_2} (p_i * (\alpha f^* + S(\alpha)^2)).\Phi_i - \frac{\alpha f^*\Phi}{n} + 
         \sum_{i \ge n\gamma_2} p_i * (\alpha f^* + S(\alpha)^2).e^{-\alpha s_i}
\end{aligned}
\end{equation}
   %The last inequality follows since $\alpha.f < 1/2$ and $\sum_{i \ge (n\gamma_2)} p_i < 1$.
  Now, we need to upper bound the term $\sum_{i < n\gamma_2} (p_i * (\alpha.f^* + (\alpha)^2).\frac{e^{\alpha.s_i}}{n^2+i})$. Since $p_i$ is non-decreasing and $\Phi_i$ is non-increasing, the maximum value is achieved when $e^{\alpha s_i} \sum_{i=1}^{n\gamma_2} \frac{1}{(n^2+i)} = \Phi$ for each $i < n\gamma_2$. Hence, $e^{\alpha s_i} = \frac{\Phi(n+\gamma_2)}{\gamma_2}$. Hence, the maximum value is given as follows.
%%%%
\begin{equation}
\begin{aligned}
   \sum_{i=1}^{n\gamma_2} p_i\Phi_i & \le \frac{\Phi(n+\gamma_2)}{\gamma_2} * \sum_{i=1}^{n\gamma_2} [ \frac{2i-1}{n^2} * \frac{1}{n^2+i} ] \\
       & \le \frac{(n+\gamma_2)\Phi}{n^2\gamma_2} * \frac{\gamma_2(2n\gamma_2 - 1)}{(n+\gamma_2)}\\
       &\le \frac{(2n\gamma_2 - 1)\Phi}{n^2}
\end{aligned}
\end{equation}
%%%%%
  Similarly, one can upper bound the term, $\sum_{i \ge n\gamma_2} (p_i \frac{e^{\alpha.s_i}}{(n^2+i)})$. Since $p_i$ is non-decreasing and $\Phi_i$ is non-increasing, the maximum value is achieved when $e^{\alpha s_i} ( \sum_{i=(n\gamma_2)}^{n} \frac{1}{n^2+i}) = \Phi_{(\ge n\gamma_2)}$ for each $i \ge n\gamma_2$. Hence, $e^{\alpha s_i} = \frac{\Phi_{(\ge n\gamma_2)}(n+\gamma_2)}{(1-\gamma_2)}$. 
%Hence, the required upper bound is given as follows.
%%%%
%\begin{equation}
%\begin{aligned}
%   \sum_{i=n\gamma_2}^{n} p_i\Phi_i & \le \frac{\Phi}{\ln(1/\gamma_2)} * \sum_{i=n\gamma_2}^{n} [ \frac{2i-1}{n^2} * \frac{1}{i} ] \\
%       & \le \frac{2(1-\gamma_2)\Phi}{n\ln(n)} \because \quad \Phi_{(\ge n\gamma_2)} \le \frac{\Phi}{ln(n)}
%\end{aligned}
%\end{equation}
%%%%%% 
  Thus, the expected change in $\Phi$ can be computed, using equation~\eqref{eq-phi-w-1} and the above bound, as follows:
\begin{equation}
\begin{aligned}
  \EE[ \Delta\Phi | x(t)] & \le (\alpha.f^* + S\alpha^2)\frac{(2n\gamma_2 - 1)\Phi}{n^2} - \frac{\alpha.f^*}{n} * \Phi + (\alpha.f^* + S\alpha^2)* \frac{(2n-1)(n+\gamma_2)\Phi_{(\ge n\gamma_2)}}{n^2(n+1)}\\ 
   & \le \frac{2\alpha.f^*\gamma_2\Phi}{n} - \frac{\alpha.f^*\Phi}{n}\\
   & \le \frac{(2\gamma_2 - 1)\alpha.f^*\Phi}{n}\\
   & \le \frac{-\alpha.f^*\Phi}{8n} \because \gamma_2 < (1/2 - 1/16)
\end{aligned}
\end{equation}
%
%
%$\Box$\\
\end{proof}
%
%
%%%%%%%%%%%%% Weighted Case: \Psi easy Case %%%%%%%%%%%%%%%
%%%%%%%%
%%%%%%%%
\begin{lemma}\label{lemma-psi-easy-case-w}
  Let $\Psi$ be defined as above. If $s_{(n\gamma_1)}(t) \ge 0$ then,
  $\EE[ \Psi(t+1) | x(t) ] \le (1-\frac{\alpha.f^*}{2n})\Psi$ 
\end{lemma}

\begin{proof}
	The proof is similar to that of Lemma~\ref{lemma-phi-easy-case-w}. See Appendix~\ref{app:2_w_proof} for details of the proof.
\end{proof}

Now, we consider the remaining cases and show that in case the load across the bins , at time $t$, is not reasonably balanced, then for $s_{n\gamma_2} > 0$, either $\Psi$ dominates $\Phi$ or $\Gamma < c/n$, where, $c = poly(1/\epsilon)$.
%%%%%
%%%%%
%%%%%%%%%%%%% Weighted Case: \Phi unbalanced %%%%%%%%%%%%%%%
%%%%%
%%%%%
\begin{lemma}\label{lemma-s-gamma-phi-w}
   Let, $s_{(n\gamma_2)} \ge 0$ and $\, \EE[\Delta\Phi|x(t)] \ge -\alpha.f^*\Phi/8n$. Then, either $\Phi < \epsilon\gamma_1 * \Psi$, or $\Gamma < \frac{c}{n}$ for some $c = poly(1/epsilon)$.
\end{lemma}
\begin{proof}
    From equation~\eqref{eq-phi-base}, we get:
\begin{equation}
\begin{aligned}
   \EE[\Delta\Phi | x(t)] & \le \sum_{i=1}^{n} (p_i * (\alpha.f^* + S\alpha^2) - \alpha.f^*/n).\Phi_i\\
    & \le \sum_{i \le n\gamma_3} (p_i * (\alpha.f^* + S\alpha^2) - \alpha.f^*/n).\Phi_i + \sum_{i > n\gamma_3} (p_i * (\alpha.f^* + S\alpha^2) - \alpha.f^*/n) * \Phi_i\\
    & \le  \frac{n\Phi_{(\le n\gamma_3)}}{\gamma_3}(\alpha.f^* + S\alpha^2)*\sum_{i \le (n\gamma_3)} \frac{2i-1}{n^2(n^2+i)} +\\ 
    & (\alpha.f^* + S\alpha^2)\frac{n\Phi_{(> n\gamma_3)}}{1-\gamma_3} * \sum_{i > (n\gamma_3)} \frac{2i-1}{n^2(n^2+i)} - \frac{\alpha f^*\Phi}{n}\\
    & \le \frac{\alpha f^*\Phi_{\le n\gamma_3}}{n} \frac{(2n\gamma_3 - 1)}{n+\gamma_3} + \frac{\alpha f^*\Phi_{> n\gamma_3}}{1-\gamma_3} \frac{\gamma_3}{n+1} - \frac{\alpha.f^*\Phi}{n}\\
    & \le \frac{\alpha f^*\Phi}{n} (\frac{2\gamma_3}{n+\gamma_3} - 1) + \alpha f^*\Phi_{>n\gamma_3} [-\frac{2n\gamma_3 - 1}{n(n+\gamma_3)} + \frac{\gamma_3}{(1-\gamma_3)(n+1)}]
\end{aligned}
\end{equation}
  Now, since $\EE[\Delta\Phi|x(t)] \ge -\alpha f\Phi/8n$, we get: $\Phi \le 4n\Phi_{(> n\gamma_3)} [\frac{(n-2)\gamma_3 + 1}{(n+1)(1-\gamma_3)(4n+3\gamma_3)}]$.
  Let, $B = \sum_{i} max(0, s_i)$. Note, $\sum_{i} s_i = 0$, since for each dimension $d$, the update maintains that, $\sum_{d=1}^{D} x_i^d(t) = 0$. Further, because, $s_{n\gamma_3} > 0$, $\Phi_{(>n\gamma_3)} \le \frac{1-\gamma_3}{n+\gamma_3} * e^{(\frac{\alpha.B}{n\gamma_3})}$. This implies that, $\Phi \le \frac{4(n-2)\gamma_3 e^{(\frac{\alpha.B}{n\gamma_3})}}{(4n+3\gamma_3)(n+\gamma_3)}$.

Since, $s_{(n\gamma_2)} > 0$, so, $\Psi \ge \frac{n}{\gamma_2} * e^{\frac{\alpha.B}{(n-n\gamma_2)}}$. If $\Phi < \epsilon\gamma_1 * \Psi$, then we are done. Else, $\Phi \ge \epsilon\gamma_1 * \Psi$. This implies:
\begin{equation}\nonumber
  e^{\frac{\alpha.B}{n\gamma_3}} \frac{4(n-2)\gamma_3}{(4n+3\gamma_3)(n+\gamma_3)} \ge \Phi \ge \epsilon\gamma_1 * \Psi \ge \frac{n\epsilon\gamma_1}{\gamma_2} * e^{\frac{\alpha.B}{(n-n\gamma_2)}}
\end{equation}
  Thus, $e^{\alpha.B/n} \le (\frac{4(n-2)\gamma_3\gamma_2}{n\epsilon\gamma_1(4n+3\gamma_3)(n+\gamma_3)})^{\frac{(1-\gamma_2)\gamma_3}{(\gamma_3 + \gamma_2 - 1)}}$. So, $\Gamma \le \frac{1 + \theta}{\epsilon} * \Phi \le \frac{1 + \theta}{\epsilon} * \frac{4(n-2)
  \gamma_3}{(4n+3\gamma_3)(n+\gamma_3)} * e^{\frac{\alpha.B}{n\gamma_3}}$ %\le \frac{1 + \theta}{\epsilon} * \frac{4\gamma_4\ln(n)}{\ln(n\gamma_3)} * (\frac{4\gamma_4}{\epsilon\gamma_1})^{\frac{\gamma_1}{(\gamma_3 - \gamma_1)}}$. 
  
Hence, $\Gamma < c/n$, where, $c = poly(1/\epsilon)$.
%
%$\Box$\\
\end{proof}

In the Lemma below, we consider the case where the load across the bins at time, $t$, is not reasonably balanced, and $s_{(n\gamma_1)} < 0$. Here, we show that either $\Phi$ dominates $\Psi$ or the potential function is less than $c/n$ for $c = poly(1/\epsilon)$.
%%%%%%%
%%%%%%%%%%%%%%%% Weighted Case: \Psi unbalanced case %%%%%%%%%%%%%%
%%%%%%%
%%%%%%%
\begin{lemma} \label{lemma-s-gamma-psi-w}
   Let, $s_{(n\gamma_1)} < 0$ and $\EE[\Delta\Psi|x(t)] \ge -\alpha.f^*\Psi/2n$. Then, either $\Psi < \epsilon\gamma_1 * \Phi$, or $\Gamma < c/n$ for some $c = poly(1/epsilon)$.
\end{lemma}

\begin{proof}
	The proof is similar to that of Lemma~\ref{lemma-s-gamma-phi-w}. See Appendix~\ref{app:3_w_proof} for details of the proof.
\end{proof}
  Now, we consider combinations of the cases considered so far and can show that the potential function, $\Gamma$, behaves as a super-martingale.
%%%%%%%
%%%%%%%
%%%%%%%%%%%%%%%% Weighted Case: Supermartingale Property %%%%%%%%%%%%%%
%%%%%%%
%%%%%%%
\begin{theorem}\label{thm-super-mart-w}
   For the potential function, $\Gamma$, $\EE[ \Gamma(t+1) | x(t)] \le (1 - \frac{\alpha.f^*}{16n(1 + \epsilon\gamma_1)})\Gamma(t) + \frac{c}{n^2}$, for constant $c = poly(1/\epsilon)$.
\end{theorem}
\begin{proof}
    We consider the following cases on intervals of values for $s_i$.
\begin{itemize}
\item \textbf{Case 1:} $s_{(n\gamma_1)} \ge 0$ and $s_{(n\gamma_2)} < 0$. Using, Lemma~\ref{lemma-phi-easy-case-w} and Lemma~\ref{lemma-psi-easy-case-w}, we can immediately see that, $\EE[ \Gamma(t+1) | x(t)] \le (1 - \alpha.f^*/8n)\Gamma(t)$ and hence, the result is also true.\\
\item \textbf{Case 2:} $s_{n\gamma_1} \ge s_{n\gamma_2} > 0$. This represents a high load imbalance across the bins. In some cases, $\Phi$ may grow but the asymmetry in the load implies that $\Gamma$ is dominated by $\Psi$. Thus, the decrease in $\Psi$ offsets the increase in $\Phi$ and hence the expected change in $\Gamma$ is negative.\\
  Specifically, if $\EE[\Delta\Phi|x] \le \frac{-\alpha f\Phi}{8n}$, then using Lemma~\ref{lemma-psi-easy-case-w} we get that $\EE[ \Gamma(t+1) | x(t)] \le (1 - \alpha f/8n)\Gamma(t)$; else we consider the following two cases:\\
\begin{itemize}
\item \textbf{Case 2a:} $\Phi < \epsilon\gamma_1 * \Psi$. Here, using Lemma~\ref{lemma-psi-easy-case-w} and Corollary~\ref{cor-phi-gen-w}, we get:
\begin{equation}\nonumber
\begin{aligned}
   \EE[\Delta\Gamma|x] & = \EE[\Delta\Phi|x] + \EE[\Delta\Psi|x]\\
        & \le \frac{\alpha.f^*}{n}.\Phi - \frac{\alpha f^*}{2n} * \Psi\\
        & \le -\frac{\epsilon}{4n}.\Psi\\
        & \le - \frac{\epsilon}{4n(1 + \epsilon\gamma_1)}\Gamma
\end{aligned}
\end{equation}
\item \textbf{Case 2b:} $\Gamma < c/n$. Here, using Corollary~\ref{cor-psi-gen-w} and Corollary~\ref{cor-phi-gen-w}, we get:
\begin{equation}\nonumber
\begin{aligned}
   \EE[\Delta\Gamma|x] & \le \frac{\alpha.f^*}{n} * \Gamma
        & \le \frac{c\alpha.f^*}{n^2}
\end{aligned}
\end{equation}
   But, $c/n^2 - ((\alpha.f^*/8n) * \Gamma) \ge c/n^2(1 - \alpha.f^*/8) \ge c/n^2(1 - \alpha.f^*/2) \ge \frac{c\alpha.f^*}{n^2}$.\\
    Hence,$\EE[\Delta\Gamma|x] \le - \frac{\alpha.f^*\Gamma}{8n} + \frac{c}{n}$.
\end{itemize}
\item \textbf{Case 3:} $s_{n\gamma_2} \le s_{n\gamma_1} < 0$. Here, if $\EE[\Delta\Psi|x] \le \frac{-\alpha.f^*}{2n}\Psi$, then using Lemma~\ref{lemma-phi-easy-case-w}, we get that $\EE[ \Gamma(t+1) | x(t)] \le (1 - \alpha.f^*/8n)\Gamma(t)$; else we consider the following two cases:
\begin{itemize}
\item \textbf{Case 3a:} $\Psi < \epsilon\gamma_1 * \Phi$. Here, using Lemma~\ref{lemma-phi-easy-case-w} and Corollary~\ref{cor-psi-gen-w}, we get:
\begin{equation}\nonumber
\begin{aligned}
   \EE[\Delta\Gamma|x] & = \EE[\Delta\Phi|x] + \EE[\Delta\Psi|x]\\
        & \le -(\alpha.f^*/8n)\Phi + \alpha.f^*/n * \Psi\\
        & \le -(\alpha.f^*/8n).\Phi + (\gamma_1\epsilon\alpha.f^*/n)\Phi\\
        & \le \frac{-\alpha.f^*}{16n}\Phi\\
        & \le \frac{-\alpha.f^*}{16n(1+\epsilon\gamma_1)}*\Gamma
\end{aligned}
\end{equation}
\item \textbf{Case 3b:} $\Gamma < c/n$. Here, using Corollary~\ref{cor-phi-gen-w} and Corollary~\ref{cor-psi-gen-w}, we get:
\begin{equation}\nonumber
\begin{aligned}
   \EE[\Delta\Gamma|x] & = \alpha.f^*/n * \Gamma
        & \le \frac{c\alpha.f^*}{n^2}
\end{aligned}
\end{equation}
  Hence, this case follows similarly as \textit{Case 2b} above. 
\end{itemize}
\end{itemize}
%
%$\Box$\\
\end{proof}
  Now, we can prove using induction that the expected value of $\Gamma$ remains bounded.
%%%%%%%
%%%%%%%
%%%%%%%%%%%% Weighted Case %%%%%%%%%%%%%%%%%%%%%%%%%%%
\begin{theorem}\label{thm-pot-up-bound-w}
  For any time $t \ge 0$, $\EE[\Gamma(t)] \le \frac{16c(1 + \epsilon\gamma_1)}{n\alpha.f^*}$
\end{theorem}
\begin{proof}
  Using induction we can prove this claim. For $t = 0$, it is trivially true since $\Gamma(0) \le 2/n$.
Using Theorem~\ref{thm-super-mart-w}, we get:
\begin{equation}\nonumber
\begin{aligned}
   \EE[\Gamma(t+1)] & = E[E[\Gamma(t+1)| \Gamma(t)]]\\
            & \le \EE[ (1 - \frac{\alpha.f^*}{16n(1+\epsilon\gamma_1)})\Gamma(t) + \frac{c}{n}]\\
            & \le \frac{16c(1 + \epsilon\gamma_1)}{n\alpha.f^*} - \frac{c}{n^2} + \frac{c}{n^2}\\
            & \le \frac{16c(1 + \epsilon\gamma_1)}{n\alpha.f^*}
\end{aligned}
\end{equation}
%
%$\Box$\\
\end{proof}
%%%%
%%%%
%%%%
\begin{theorem}\textbf{Variable $f$ Case (Weighted Case) Gap:}
  Using the bias in the probability distribution in favor of lightly loaded md-bins as obtained from the $d$-choice process, and assuming that in each ball, each dimension is chosen as $1$ with probability $q$ (variable $f$ case); the expected and probabilistic upper bound on the gap (maximum dimensional gap) across the multidimensional bins is given as follows. Let, $\delta = \frac{16c(1 + \epsilon\gamma_1)}{\alpha.f*}$, and $\zeta > 0$, then:
\begin{equation}\nonumber
\begin{aligned}
    E[Gap(t)] & \le 2q\log(n)/\epsilon + 2q\log(\delta)/\epsilon\\
    Pr[Gap(t) & > (mq/n)^{1/2+\zeta}(4q\log(n)/\epsilon + 4q\log(\delta)/\epsilon]) \le \frac{1}{qm}\\
\end{aligned}
\end{equation}
\end{theorem}
\begin{proof}
%%%%%
Since, each dimension is assigned $1$ with probability $q$, the average number of ones per md-ball is $f^* = Dq$.
%%%%%
Let, $a$ be the winning md-bin and $m$ be the winning dimension that represents $Gap(t)$. The number of ones in any ball, $f$, follows a Binomial($D$, $q$) distribution and has finite second moment. Using the analysis for the weighted balls case, we get, $\EE[ e^{\alpha.s_a} ] \le n\delta$, where $\alpha.f^* \le \epsilon/2$. So, 
$\EE[ e^{\alpha.(x_a^m + \sum_{d \ne m} x_a^d)} ] \le n\delta$. 
  Taking, logarithm of both sides, we get:
\begin{equation}\label{eq:sum-bound-w}
\begin{aligned}
    E[x_a^m] + \sum_{d \ne m} E[x_a^d] & \le \log(n)/\alpha + \log(\delta)/\alpha\\
%    \Rightarrow E[l_a^m] + \sum_{d \ne m} E[l_a^d]  - (D * \frac{mq}{n})& \le \log\log(n)/\alpha + \log\log(\delta)/\alpha
\end{aligned}
\end{equation}
%
%  In the second inequality, $l_a^d$, represents the load in dimension $d$ for bin $a$. Further, since the average load in each dimension is $\frac{mq}{n}$. 
If $k$ is the expected number of balls were thrown in bin $a$ minus the average number of balls per bin, then $E[x_a^m] = kq$ and similarly, $\sum_{d \ne m} E[x_a^d] = (D-1)kq$. Hence, we get:
\begin{equation}\nonumber
\begin{aligned}
    Dkq & \le 2f^*\log(n)/\epsilon + 2f^*\log(\delta)/\epsilon\\
    \Rightarrow k & \le 2\log(n)/\epsilon + 2\log(\delta)/\epsilon\\
    \Rightarrow E[x_a^m] & \le 2q\log(n)/\epsilon + 2q\log(\delta)/\epsilon\\
\end{aligned}
\end{equation}
  The probabilistic bound can be computed similar to the fixed $f$ case (Theorem~\ref{thm-fixed-f}) using the Chernoff bound.
%Now, the $Pr[s_a > 4f\log(n)/\epsilon + 2f/\epsilon * \log(\delta)] \le
%Pr[ \Gamma(t) \ge nE[\Gamma(t)]] \le 1/n$ (using Markov's Inequality); where $s_a = \sum_{d=1}^{D} x_a^d$. Hence, $Pr[x_a^m > 2* [2f\log(n)/\epsilon + m/n(1 - f/D) + 2f/\epsilon * \log(\delta)] \le 1/n$. 
%
$\Box$
\end{proof}
%%%%
%%%%
%%%%
  Note that the for the scalar case, when the expected weight of the distribution is $W^*$, the upper bound on the gap obtained is $O(W^*\log(n))$, which after normalization, i.e. $E(W) = W^* = 1$, leads to $O(\log(n))$ gap. This improves upon the best prior known bound of $O(n^c)$ given in~\cite{kunal-weighted}. 
\section{$(1+\beta$) Choice Process with Multidimensional Balls and Bins}
%%%%%
\par In this section we present upper and lower bounds on the gap for the $(1+\beta)$ choice process with multidimensional balls and bins.
%%%%%
%%%%%
\subsection{Markov Chain Specification}\label{sec:markov-chain-beta}
\label{sec:markov}

As mentioned earlier, a balls-and-bins process can be characterized by a probability distribution vector $(p_1, p_2, p_3,...p_n)$, where, $p_i$ is the probability a ball is placed in the $i^{th}$ most loaded multidimensional bin. Let $x_i^d(t)$ be the random variable, that denotes the \textit{weight in dimension $d$ for bin $i$} and is equal to the load of the $d^{th}$ dimension of the $i^{th}$ bin minus the average load in dimension $d$. So, $\sum_{i=1}^{n} x_i^d(t) = 0, \forall d \in [1..D]$. Let, $s_i(t)$ denote the sum of the loads (minus corresponding dimension averages) across all $D$ dimensions for the bin $i$ at time $t$, expressed as $s_i(t) = \sum_{d=1}^D x_i^d$. It is assumed that bins are sorted by $s_i(t)$. So, $s_i \ge s_{i+1} \forall i \in [1..n-1]$. The process defines a Markov chain over the matrices, $x(t)$ as follows:

\begin{itemize}
\item Sample $j \in_p [n]$.
\item Set $r_i = s_i(t) + f(1 - 1/n)$, for $i = j$. Since, each md-ball has $f$ non-zero entries , so each of these $f$ dimensions in the bin, $i$, will be incremented by $1 - 1/n$.
\item Set $r_i = s_i(t) - f/n$, for $i \ne j$. Since, each md-ball has $f$ non-zero entries, so the each of the corresponding $f$ dimensions in the bin, $i$, will be decremented by $1/n$. This ensures that for each dimension the sum across all the bins is $0$.
\item Obtain $s(t+1)$ by sorting $r(t)$.
\end{itemize}

Fig.~\ref{fig:md-balls-bins} (in the Appendix~\ref{app:fig}) illustrates a multidimensional balls and bins scenario. The bounds on the gap will be proven for a family of probability distribution vectors $p$. As mentioned earlier, he md-bins are sorted based on their total dimensional load, i.e. sum of the weights across all dimensions for each bin ($s_i$ for bin $i$). 
We make the following assumptions:
\begin{itemize}
\item %% Ass(1) %%%
  $\forall i \in [1,n-1], p_i \le p_{i+1}$
  This assumption states that the allocation rule is no worse than the $1$-choice scheme.
\item %% Ass(2) %%%
  For some constants, $\epsilon > 0$, $\theta > 1$ and $0 < \gamma_3 < \gamma_4 < 1$, where $\gamma_3 + \gamma_4 = 1$, it holds that:
\begin{equation}\label{ass-2}
  p_{(n\gamma_3)} \le \frac{(1 - \theta\epsilon)}{n},\quad \text{and},\quad p_{(n\gamma_4)} \ge \frac{(1 + \theta\epsilon)}{n}
\end{equation}
  This assumption states that the allocation rule strictly prefers the least loaded $\gamma_3$ fraction of the $n$ bins over the most loaded $(1 - \gamma_4)$ fraction. 
\end{itemize}

These assumptions imply that for some constants, $\gamma_1$ and $\gamma_2$, where, $0 < \gamma_1 < \gamma_3 < 1/2 < \gamma_4 < \gamma_2 < 1$ and $\theta\gamma_1 = 1$; $\gamma_1 + \gamma_2 = 1$, we have the following:\\ $\sum_{i \ge (n\gamma_2)} p_i \ge (\gamma_1 + \epsilon)$ and $\sum_{i \le (n\gamma_1)} p_i \le (\gamma_1 - \epsilon)$. This will be useful in the proof. Note that the $(1+\beta)$ choice process satisfies these assumptions for $\epsilon = \beta(1 -2\gamma_3)/\theta$, since $p_{(n\gamma_3)} \le (1 - \beta)/n + 2(n\gamma_3 - 1)\beta/n^2 \le (1 - \beta(1-2\gamma_3))/n$, and similarly $p_{(n\gamma_4)} \ge (1 + \beta(2\gamma_4 - 1))/n$.

In the remaining analysis, we assume that when an md-ball arrives, then the selection of the bins is based on $s_i$, i.e. total sum of weights across all dimensions for the randomly selected bins (Fig.~\ref{fig:md-balls-bins} in Appendix~\ref{app:fig}). In particular, for the $(1+\beta)$ choice process, when two bins are randomly selected (with $\beta$ probability), the md-ball (with $f$ non-zero entries) is assigned to the md-bin with the lowest $s_i$. Using this selection mechanism, we prove the upper and lower bounds on the gap obtained for the $(1+\beta)$ choice process. Note that, this is a different allocation mechanism than that considered in~\cite{md-mm} where the \textit{max} objective is considered over the restricted set of $f$ populated dimensions in the current md-ball.

\subsection{Upper Bound On the Gap}\label{sec:upper-bound-beta}
  We assume that $\epsilon \le 1/4$. Further, let $\alpha = \epsilon/2f$. Define the following potential functions:
\begin{equation}
\begin{aligned}
     \Phi(t) & = \Phi(s(t)) = \sum_{i=1}^{n} e^{\alpha.s_i}\\
     \Psi(t) & = \Psi(s(t)) = \sum_{i=1}^{n} e^{-\alpha.s_i}\\
     \Gamma(t) & = \Gamma(s(t)) = \Phi(t) + \Psi(t)     
\end{aligned}
\end{equation}
  where, $s_i(t) = \sum_{d=1}^{D} x_i^d(t)$

In the beginning, each dimension for each bin has $0$ weight, thus $s_i = 0, \forall  i$ and hence, $\Gamma(0) = 2n$. We show that if $\Gamma(x(t)) \ge an$ for some $a > 0$, then $\mathbb{E}[\Gamma(t+1) | x(t)] \le (1 - \frac{\epsilon^2(1 - 2\gamma_1)}{4n(1 + \epsilon\gamma_1)}) * \Gamma(t)$. This helps in demonstrating that for every given $t$, $\mathbb{E}[\Gamma(t)] \in O(n)$. This implies that the maximum gap is $O(\log (n))$ w.h.p.

First, consider the change in $\Phi(t)$ (also refers to $\Phi$ by default) and $\Psi(t)$ (also refers to $\Psi$ by default) separately when a ball is thrown with the given probability distribution. 
\begin{lemma}
  When an md-ball is thrown into an md-bin, the following inequality holds:
\begin{equation}
   \EE[ \Phi(t+1) - \Phi(t) | x(t) ] \le \sum_{i=1}^{n} [ p_i * (\alpha.f + (\alpha.f)^2)  - \alpha.f/n ]. e^{\alpha.s_i} 
\end{equation}
\end{lemma}
\begin{proof}
  Let $\Delta_i$ be the expected change in $\Phi$ if the ball is put in bin, $i$. So, $r_i(t+1) = s_i + f(1-1/n)$; and for $j \ne i$, $r_j(t+1) = s_j(t) - f/n$. The new values i.e. $s(t+1)$ are obtained by sorting $r(t+1)$ and $\Phi(s) = \Phi(r)$. The expected contribution of bin, $i$, to $\Delta_i$ is given as follows:
\begin{equation}\nonumber
\begin{aligned}
   \EE[ e^{\alpha.(s_i + f(1-1/n))} ] - e^{\alpha.s_i}
   & = e^{\alpha.s_i} [ e^{\alpha.f(1-1/n)} - 1 ]
\end{aligned}
\end{equation}
 Similarly, the expected contribution of bin, $j$ ($j \ne i$) to $\Delta_i$ is given as:
\begin{equation}\nonumber
\begin{aligned}
   \EE[ e^{\alpha.(s_j - f/n)} ] - e^{\alpha.s_j}
   & = e^{\alpha.s_j} [ e^{-\alpha.f/n} - 1 ]
\end{aligned}
\end{equation}
 Therefore, $\Delta_i$ is given as follows:
\begin{equation}\nonumber
\begin{aligned}
   \Delta_i & = e^{\alpha.s_i} [ e^{\alpha.f(1-1/n)} - 1] \, + \,
              \sum_{j \ne i} e^{\alpha.s_j} (e^{-\alpha.f/n} - 1)\\
   & = e^{\alpha(s_i- f/n)} (e^{\alpha.f} - 1) + (e^{-\alpha.f/n} - 1).\Phi
\end{aligned}
\end{equation}
  Thus, we get the overall expected change in $\Phi$ as follows:
\begin{equation}
\begin{aligned}
   \EE[ \Phi(t+1) - \Phi(t) | x(t) ] & = \sum_{i=1}^{n} p_i * \Delta_i\\
   & = \sum_{i=1}^{n} p_i * [ e^{\alpha(s_i- f/n)} (e^{\alpha.f} - 1) + (e^{-\alpha.f/n} - 1).\Phi ]\\
   & = \sum_{i=1}^{n} p_i * e^{\alpha(s_i- f/n)} (e^{\alpha.f} - 1) +                        (e^{-\alpha.f/n} - 1).p_i.\Phi\\
   & = \sum_{i=1}^{n} [ p_i * e^{-\alpha.f/n} (e^{\alpha.f} - 1) + (e^{-\alpha.f/n} - 1) ]. e^{\alpha.s_i}
\end{aligned}
\end{equation}
 Now, $e^{(-\alpha.f/n)} * (e^{\alpha.f} - 1)$ can be approximated as follows:
\begin{equation}\nonumber
\begin{aligned}
   e^{(-\alpha.f/n)}.(e^{\alpha.f} - 1) & \le
   (1 - \alpha.f/n + (\alpha.f/n)^2) * (1 + \alpha.f + (\alpha.f)^2 - 1)\\
   & \sim \alpha.f + (\alpha.f)^2 + O((\alpha.f)^2/n)\\   
   e^{(-\alpha.f/n)}.(e^{\alpha.f} - 1) & \lessapprox (\alpha.f + (\alpha.f)^2)
\end{aligned}
\end{equation}
Above, since, $(\alpha.f)^2/n$ is very small for large $n$, we have ignored the small terms. Similarly, $(e^{-\alpha.f/n} - 1) \lessapprox -\alpha.f/n$
Hence, the expected change in $\Phi$ can be given by:
\begin{equation}\label{eq-phi-base-beta}
\begin{aligned}
   \EE[ \Phi(t+1) - \Phi(t) | x(t) ]
   \le \sum_{i=1}^{n} [ p_i * (\alpha.f + (\alpha.f)^2)  - \alpha.f/n ]. e^{\alpha.s_i}
\end{aligned}
\end{equation}
%
%$\Box$\\
\end{proof}
 Simplifying further and observing that $\Phi_i$ decreases and $p_i$ increases with increasing $i$ from $1$ to $n$, one gets the following Corollary.
\begin{corollary}\label{cor-phi-gen-beta}
   $\EE[ \Phi(t+1) - \Phi(t) | x(t) ] \le (\alpha.f)^2 * \Phi / n$ 
\end{corollary}
\begin{proof}
   Since, $p_i$ are increasing and $\Phi_i$ are decreasing, the maximum value taken by RHS of equation~\eqref{eq-phi-base-beta} will be when $p_i = 1/n$ for all $i \in [1..n]$. Simplifying, we get the result.
%
%$\Box$\\
\end{proof}

  Similarly, the change in $\Psi$ can be derived as follows.
\begin{lemma}
  When an md-ball is thrown into an md-bin, the following inequality holds:
\begin{equation}\label{eq-psi-base-beta}
   \EE[ \Psi(t+1) - \Psi(t) | x(t) ] \le
   \sum_{i=1}^{n} [ p_i * (-\alpha.f + (\alpha.f)^2)  + \alpha.f/n ]. e^{-\alpha.s_i} 
\end{equation}
\end{lemma}
%%%% REM %%% : proof removed for brevity
%
Further observing that $p_i > 0$, one gets the following Corollary.
\begin{corollary}\label{cor-psi-gen-beta}
   $\EE[ \Psi(t+1) - \Psi(t) | x(t) ] \le (\alpha.f.\Psi)/ n$ 
\end{corollary}

In the next two lemmas, Lemma~\ref{lemma-phi-easy-case-beta} and Lemma~\ref{lemma-psi-easy-case-beta}, we consider a reasonably balanced md-bins scenario. We show that for such cases, the expected potential decreases. Specifically, for $s_{(n\gamma_2)} \le 0$, the expected value of $\Phi$ decreases and for $s_{(n\gamma_1)} \ge 0$, the expected value of $\Psi$ decreases.
\begin{lemma}\label{lemma-phi-easy-case-beta}
  Let $\Phi$ be defined as above. If $s_{(n\gamma_2)}(t) \le 0$ then, $\EE[ \Phi(t+1) | x(t) ] \le (1-\frac{\epsilon^2}{4n})\Phi + 1$ 
\end{lemma}

\begin{proof}
    From equation~\eqref{eq-phi-base-beta}, we get,
\begin{equation}\label{eq-phi-1-beta}
\begin{aligned}
  \EE[ \Phi(t+1) - \Phi(t) | x(t)] & \le \sum_{i=1}^{n} (p_i * (\alpha f + (\alpha f)^2) - \alpha f/n). e^{\alpha s_i}\\
   & \le \sum_{i < n\gamma_2} (p_i * (\alpha f + (\alpha f)^2) - \alpha f/n).e^{\alpha s_i} + 
         \sum_{i \ge n\gamma_2} p_i * (\alpha f + (\alpha f)^2).e^0\\
   & \le \sum_{i < n\gamma_2} (p_i * (\alpha f + (\alpha f)^2) - \alpha f/n).e^{\alpha s_i} + 1
\end{aligned}
\end{equation}
   The last inequality follows since $\alpha.f < 1/2$ and $\sum_{i \ge (n\gamma_2)} p_i < 1$.
  Now, we need to upper bound the term $\sum_{i < n\gamma_2} (p_i * (\alpha.f + (\alpha.f)^2).e^{\alpha.s_i})$. Since $p_i$ is non-decreasing and $\Phi_i$ is non-increasing, the maximum value is achieved when $\Phi_i = (\Phi/(n\gamma_2))$ for each $i < n\gamma_2$. Hence, the maximum value is: $(\alpha.f + (\alpha.f)^2)(\gamma_2 - \epsilon)\Phi/(n\gamma_2)$.
  Thus, the expected change in $\Phi$ can be computed, using equation~\eqref{eq-phi-1-beta} and the above bound, as follows:
\begin{equation}
\begin{aligned}
  \EE[ \Phi(t+1) - \Phi(t) | x(t)] & \le (\alpha.f + (\alpha.f)^2)(\gamma_2 - \epsilon)\Phi/(n\gamma_2) - \alpha.f/n * \Phi + 1\\ 
   & \le (\alpha.f)^2.\Phi/n - (\alpha f\epsilon\Phi/n\gamma_2) + 1\\
   & \le \epsilon^2.\Phi/4n -\epsilon^2\Phi/(2n\gamma_2) + 1\\
   & \le \frac{-\epsilon^2}{4n}\Phi + 1
\end{aligned}
\end{equation}
%
%
%$\Box$\\
\end{proof}
\begin{lemma}\label{lemma-psi-easy-case-beta}
  Let $\Psi$ be defined as above. If $s_{(n\gamma_1)}(t) \ge 0$ then,
  $\EE[ \Psi(t+1) | x(t) ] \le (1-\frac{\epsilon^2}{4n})\Psi + 1$ 
\end{lemma}

\begin{proof}
	The proof is similar to that of Lemma~\ref{lemma-phi-easy-case-beta}. See Appendix~\ref{app:1_proof_beta} for details of the proof.
\end{proof}

Now, we consider the remaining cases and show that in case the load across the bins , at time $t$, is not reasonably balanced, then for $s_{n\gamma_2} > 0$, either $\Psi$ dominates $\Phi$ or the potential function is $O(n)$.
\begin{lemma}
\label{lemma-s-gamma-phi-beta}
   Let, $s_{(n\gamma_2)} > 0$ and $\, \EE[\Delta\Phi|x(t)] \ge -\epsilon^2\Phi/4n$. Then, either $\Phi < \epsilon\gamma_1 * \Psi$, or $\Gamma < cn$ for some $c = poly(1/epsilon)$.
\end{lemma}
\begin{proof}
    From equation~\eqref{eq-phi-base-beta}, we get:
\begin{equation}
\begin{aligned}
   \EE[\Delta\Phi | x(t)] & \le \sum_{i=1}^{n} (p_i * (\alpha.f + (\alpha.f)^2) - \alpha.f/n).e^{\alpha.s_i}\\
    & \le \sum_{i \le n\gamma_3} (p_i * (\alpha.f + (\alpha.f)^2) - \alpha.f/n).e^{\alpha.s_i} + \sum_{i > n\gamma_3} (p_i * (\alpha.f + (\alpha.f)^2) - \alpha.f/n) * e^{\alpha.s_i}\\
    & \le  [((1 - \theta\epsilon)/n) * (\alpha.f + (\alpha.f)^2) - \alpha.f/n)].\Phi_{(\le n\gamma_3)} + (\alpha.f)^2.\Phi_{(> n\gamma_3))}/n\\
    & \le [-\epsilon^2/(2n\gamma_1) + \epsilon^2/4n].\Phi_{(\le n\gamma_3)} + \epsilon^2.\Phi_{(> n\gamma_3)}/4n\\
    & \le [-\epsilon^2/(2n\gamma_1) + \epsilon^2/4n].\Phi + \epsilon^2\theta\Phi_{(> n\gamma_3)}/2n \quad \because \theta\gamma_1 = 1
\end{aligned}
\end{equation}
  Now, since $\EE[\Delta\Phi|x(t)] \ge -\epsilon^2\Phi/4n$, we get: $\Phi \le \Phi_{(> n\gamma_3)}/\gamma_2$.
  Let, $B = \sum_{i} max(0, s_i)$. Note, $\sum_{i} s_i = 0$, since for each dimension $d$, the update maintains that, $\sum_{d=1}^{D} x_i^d(t) = 0$. One can observe that, $\Phi_{(>n\gamma_3)} \le n\gamma_4 * e^{\alpha.B/(n\gamma_3)}$, since $\gamma_3 + \gamma_4 = 1$. This implies that, $\Phi \le ((n\gamma_4)/(\gamma_2)) * e^{\alpha.B/(n\gamma_3)}$.

Since, $s_{(n\gamma_2)} > 0$, so, $\Psi \ge n\gamma_1 * e^{\alpha.B/(n\gamma_1)}$. If $\Phi < \epsilon\gamma_1 * \Psi$, then we are done. Else, $\Phi \ge \epsilon\gamma_1 * \Psi$. This implies:
\begin{equation}\nonumber
  (n\gamma_4/\gamma_2) * e^{\alpha B/(n\gamma_3)} \ge \Phi \ge \epsilon\gamma_1 * \Psi \ge \epsilon.n\gamma_1^2 * e^{\alpha B/(n\gamma_1)}
\end{equation}
  Thus, $e^{\alpha.B/n} \le (\frac{\theta^2\gamma_4}{\epsilon\gamma_2})^{\frac{\gamma_1\gamma_3}{(\gamma_3 - \gamma_1)}}$. So, $\Gamma \le ((1 + \theta)/\epsilon) * \Phi \le ((1 + \theta)/\epsilon) * (n\gamma_4/\gamma_2) * e^{\alpha.B/(n\gamma_3)} \le ((1 + \theta)/\epsilon) * (n\gamma_4/\gamma_2) * (\frac{\theta^2\gamma_4}{\epsilon\gamma_2})^{\frac{\gamma_1}{(\gamma_3 - \gamma_1)}}$. Hence, $\Gamma \le cn$, where, $c = poly(1/\epsilon)$.
%
%$\Box$\\
\end{proof}

In the Lemma below, we consider the case where the load across the bins at time, $t$, is not reasonably balanced, and $s_{(n\gamma_1)} < 0$. Here, we show that either $\Phi$ dominates $\Psi$ or the potential function is $O(n)$.
\begin{lemma}
\label{lemma-s-gamma-psi-beta}
   Let, $s_{(n\gamma_1)} < 0$ and $\EE[\Delta\Psi|x(t)] \ge -\epsilon^2\Phi/4n$. Then, either $\Psi < \epsilon\gamma_1 * \Phi$, or $\Gamma < cn$ for some $c = poly(1/epsilon)$.
\end{lemma}

\begin{proof}
	The proof is similar to that of Lemma~\ref{lemma-s-gamma-phi-beta}. See Appendix~\ref{app:2_proof_beta} for details of the proof.
\end{proof}
  Now, we consider combinations of the cases considered so far and can show that the potential function, $\Gamma$, behaves as a super-martingale.
\begin{theorem}\label{thm-super-mart-beta}
   For the potential function, $\Gamma$, $\EE[ \Gamma(t+1) | x(t)] \le (1 - \frac{\epsilon^2(1 - 2\gamma_1)}{4n(1 + \epsilon\gamma_1)})\Gamma(t) + c$, for constant $c = poly(1/\epsilon)$.
\end{theorem}
\begin{proof}
    We consider the following cases on intervals of values for $s_i$.
\begin{itemize}
\item \textbf{Case 1:} $s_{(n\gamma_1)} \ge 0$ and $s_{(n\gamma_2)} \le 0$. Using, Lemma~\ref{lemma-phi-easy-case-beta} and Lemma~\ref{lemma-psi-easy-case-beta}, we can immediately see that, $\EE[ \Gamma(t+1) | x(t)] \le (1 - \epsilon^2/4n)\Gamma(t) + c$, for constant $c = poly(1/\epsilon)$ and hence, the result is also true.\\
\item \textbf{Case 2:} $s_{n\gamma_1} \ge s_{n\gamma_2} > 0$. This represents a high load imbalance across the bins. In some cases, $\Phi$ may grow but the asymmetry in the load implies that $\Gamma$ is dominated by $\Psi$. Thus, the decrease in $\Psi$ offsets the increase in $\Phi$ and hence the expected change in $\Gamma$ is negative.\\
  Specifically, if $\EE[\Delta\Phi|x] \le -\epsilon^2/4n * \Phi$, then using Lemma~\ref{lemma-psi-easy-case-beta} we get that $\EE[ \Gamma(t+1) | x(t)] \le (1 - \epsilon^2/4n)\Gamma(t) + c$; else we consider the following two cases:\\
\begin{itemize}
\item \textbf{Case 2a:} $\Phi < \epsilon\gamma_1 * \Psi$. Here, using Lemma~\ref{lemma-psi-easy-case-beta} and Corollary~\ref{cor-phi-gen-beta}, we get:
\begin{equation}\nonumber
\begin{aligned}
   \EE[\Delta\Gamma|x] & = \EE[\Delta\Phi|x] + \EE[\Delta\Psi|x]\\
        & \le \frac{(\alpha.f)^2}{n}.\Phi - \frac{\epsilon^2}{4n} * \Psi + 1\\
        & \le -(1 - \epsilon\gamma_1)* \frac{\epsilon^2}{4n}.\Psi + 1
        & \le - \frac{\epsilon^2(1 - \epsilon\gamma_1)}{4n(1 + \epsilon\gamma_1)}\Gamma + 1
\end{aligned}
\end{equation}
\item \textbf{Case 2b:} $\Gamma < cn$. Here, using Corollary~\ref{cor-psi-gen-beta} and Corollary~\ref{cor-phi-gen-beta}, we get:
\begin{equation}\nonumber
\begin{aligned}
   \EE[\Delta\Gamma|x] & \le \alpha.f/n * \Gamma
        & \le c\alpha.f
\end{aligned}
\end{equation}
   But, $c - ((\epsilon^2(1 - \epsilon\gamma_1)/4n) * \Gamma) \ge c(1 - \epsilon^2(1 - \epsilon\gamma_1)/4) \ge c(1 - \epsilon/2) \ge c\alpha.f$. Hence, the result follows.
\end{itemize}
\item \textbf{Case 3:} $s_{n\gamma_2} \le s_{n\gamma_1} < 0$. Here, if $\EE[\Delta\Psi|x] \le \frac{-\epsilon^2/}{4n}\Psi$, then using Lemma~\ref{lemma-phi-easy-case-beta}, we get that $\EE[ \Gamma(t+1) | x(t)] \le (1 - \epsilon^2/4n)\Gamma(t) + c$; else we consider the following two cases:
\begin{itemize}
\item \textbf{Case 3a:} $\Psi < \epsilon\gamma_1 * \Phi$. Here, using Lemma~\ref{lemma-phi-easy-case-beta} and Corollary~\ref{cor-psi-gen-beta}, we get:
\begin{equation}\nonumber
\begin{aligned}
   \EE[\Delta\Gamma|x] & = \EE[\Delta\Phi|x] + \EE[\Delta\Psi|x]\\
        & \le -(\epsilon^2/4n)\Phi + \alpha.f/n * \Psi + 1\\
        & \le -(\epsilon^2/4n).\Phi + (\gamma_1\epsilon^2/2n)\Phi\\
        & \le \frac{\epsilon^2(\gamma_1 - 1/2)}{2n}\Phi + 1
        & \le \frac{-\epsilon^2(1 - 2\gamma_1)}{(4n(1+\epsilon\gamma_1))}*\Gamma + 1
\end{aligned}
\end{equation}
\item \textbf{Case 3b:} $\Gamma < cn$. Here, using Corollary~\ref{cor-phi-gen-beta} and Corollary~\ref{cor-psi-gen-beta}, we get:
\begin{equation}\nonumber
\begin{aligned}
   \EE[\Delta\Gamma|x] & = \alpha.f/n * \Gamma
        & \le c\alpha.f
\end{aligned}
\end{equation}
  Hence, this case follows similarly as \textit{Case 2b} above. 
\end{itemize}
\end{itemize}
%
%$\Box$\\
\end{proof}
  Now, we can prove using induction that the expected value of $\Gamma$ remains bounded.
\begin{theorem}\label{thm-pot-up-bound-beta}
  For any time $t \ge 0$, $\EE[\Gamma(t)] \le \frac{4c(1 + \epsilon\gamma_1)}{\epsilon^2(1-2\gamma_1)}n$
\end{theorem}
\begin{proof}
  Using induction we can prove this claim. For $t = 0$, it is trivially true since $\Gamma(0) = 2n$.
Using Theorem~\ref{thm-super-mart-beta}, we get:
\begin{equation}\nonumber
\begin{aligned}
   \EE[\Gamma(t+1)] & = E[E[\Gamma(t+1)| \Gamma(t)]]\\
            & \le \EE[ (1 - \frac{\epsilon^2(1-2\gamma_1)}{4n(1+\epsilon\gamma_1)})\Gamma(t) + c]\\
            & \le \frac{4c(1 + \epsilon\gamma_1)}{\epsilon^2(1-2\gamma_1)}n - c + c
\end{aligned}
\end{equation}
%
%$\Box$\\
\end{proof}
 Now, we can upper bound the gap across all the $D$ dimensions across all the $n$ md-bins. This gap is defined as follows:
\begin{equation}
    Gap(t) = \max_{d=1}^{D} [ \max_{i=1}^{n} x_i^d ]
\end{equation}
%
%%%%%%%%%%%%%%%%
%%%%%%%% For Fixed "`f" Case, m>>n %%%%%%%
%%%%%%%%%%% sub-cases: uniform(f/D) and non-uniform(prob per dim >= c)
%%%%%%%%%%%%%%%%
%%%%%%%%%%%%%
\begin{theorem}\textbf{Fixed $f$ Case:}\label{thm-fixed-f-beta}
  Using the bias ($p_{n\gamma_3} \le (1 - \theta\epsilon)/n$ and $p_{n\gamma_4} \ge (1 + \theta\epsilon)/n$) in the probability distribution in favor of lightly loaded md-bins, and assuming that $f$ dimensions are exactly populated in each md-ball with uniform distribution of $f$ dimensions over $D$, then the expected and probabilistic upper bound on the gap (maximum dimensional gap) across the multidimensional bins is given as follows. Let, $\delta = \frac{4c(1 + \epsilon\gamma_1)}{\epsilon^2(1-2\gamma_1)}$, then:
\begin{equation}\nonumber
\begin{aligned}
    E[Gap(t)] & \le 2\log(n)/\epsilon + 2f\log(\delta)/\epsilon\\
    Pr[Gap(t) & > 4\log(n)/\epsilon + 4\log(\delta)/\epsilon] \le 1/n\\
\end{aligned}
\end{equation}
\end{theorem}
\begin{proof}

Let, $a$ be the winning md-bin and $m$ be the winning dimension that represents $Gap(t)$. Now,
from Theorem~\ref{thm-pot-up-bound-beta}, we get, $\EE[ e^{\alpha.s_a} ] \le n\delta$. So, 
$\EE[ e^{\alpha.(x_a^m + \sum_{d \ne m} x_a^d)} ] \le n\delta$. 
  Let, $y_a$ denote the gap as measured by the number of md-balls in bin $a$ minus the average number of balls across the bins. Then,
\begin{equation}\label{eq:sum-bound}
\begin{aligned}
    E[s_a] & \le 1/\alpha * \log(n) + 1/\alpha * \log(\delta)\\
    \Rightarrow E[s_a] & \le 2f\log(n)/\epsilon + O(2f\log(\delta)/\epsilon)\\
    \Rightarrow f.E[y_a] & \le 2f\log(n)/\epsilon + O(2f\log(\delta)/\epsilon)\\
    \Rightarrow E[y_a] & \le 2\log(n)/\epsilon + O(2\log(\delta)/\epsilon)
\end{aligned}
\end{equation}
  The third inequality uses the fact that each ball has exactly $f$ populated dimensions. Since the $f$ dimensions are chosen uniformly and randomly from $D$ dimensions, the expected gap in any dimension (and hence the winning dimension with the maximum gap) is bounded by $O(\frac{\log(n)}{\beta})$, since $\epsilon = \Theta(\beta)$ (section~\ref{sec:markov-chain-beta}).
  %In the case of non-uniform distribution, where we assume that each dimension is chosen with probability at most $\kappa_2$ in each md-ball, the gap is bounded by $O(\frac{\kappa_2\log(n)}{\beta})$.

  %Now, the maximum value of RHS is attained when $|x_a^d|$ is the maximum. Note, that $|x_a^d| \le x_a^m$, since for each dimension $d$, $|x_a^d| < x_b^d$, where $b$ is the md-bin that has the maximum value for dimension $d$, over all the bins. By definition of $x_a^m$, it represents maximum gap across all bins and all dimensions, so $x_a^d \le x_a^m$. Hence, we get:
%
Now, the $Pr[s_a > 4f\log(n)/\epsilon + 2f/\epsilon * \log(\delta)] \le
Pr[ \Gamma(t) \ge nE[\Gamma(t)]] \le 1/n$ (using Markov's Inequality); where $s_a = \sum_{d=1}^{D} x_a^d$. Hence, $Pr[y_a > 4\log(n)/\epsilon + 4\log(\delta)/\epsilon] \le 1/n$. 
%
%$\Box$
\end{proof}
%%%%
%%%%
%%%%%%
%%%%%%
\begin{theorem}\textbf{Variable $f$ Case:}
  Using the bias ($p_{n\gamma_3} \le (1 - \theta\epsilon)/n$ and $p_{n\gamma_4} \ge (1 + \theta\epsilon)/n$) in the probability distribution in favor of lightly loaded md-bins, and assuming that each dimension is chosen as $1$ with probability $q$ (non-fixed $f$ case); the expected and probabilistic upper bound on the gap (maximum dimensional gap) across the multidimensional bins is given as follows. Let, $\delta = \frac{4c(1 + \epsilon\gamma_1)}{\epsilon^2(1-2\gamma_1)}$, then:
\begin{equation}\nonumber
\begin{aligned}
    E[Gap(t)] & \le 2\log(n)/\epsilon + m/n(1-q) + 2\log(\delta)/\epsilon\\
    Pr[Gap(t) & > 4\log(n)/\epsilon + m/n(1-q) + 4\log(\delta)/\epsilon] \le 1/n\\
\end{aligned}
\end{equation}
\end{theorem}
\begin{proof}
%%%%%
Since, each dimension is assigned $1$ with probability $q$, the average number of ones per md-ball is $f^* = Dq$.
%%%%%
Let, $a$ be the winning md-bin and $m$ be the winning dimension that represents $Gap(t)$. The number of ones in any ball, $f$, follows a Binomial($D$, $q$) distribution and has finite second moment. Using the analysis similar as for Theorem~\ref{thm-pot-up-bound-beta}, we can get (proof omitted for brevity), $\EE[ e^{\alpha.s_a} ] \le n\delta$, where $\alpha.f^* \le \epsilon/2$. So, 
$\EE[ e^{\alpha.(x_a^m + \sum_{d \ne m} x_a^d)} ] \le n\delta$. 
  Taking, logarithm of both sides, we get:
\begin{equation}\label{eq:sum-bound}
\begin{aligned}
    E[x_a^m] + \sum_{d \ne m} E[x_a^d] & \le \log(n)/\alpha + \log(\delta)/\alpha\\
    \Rightarrow E[l_a^m] + \sum_{d \ne m} E[l_a^d]  - (D * \frac{mq}{n})& \le \log(n)/\alpha + \log(\delta)/\alpha
\end{aligned}
\end{equation}
  In the second inequality, $l_a^d$, represents the load in dimension $d$ for bin $a$. Further, since the average load in each dimension is $\frac{mq}{n}$. If $k$ balls were thrown in bin $a$, then $E[l_a^m] = kq$ and similarly, $\sum_{d \ne m} E[l_a^d] = (D-1)kq$. Hence, we get:
\begin{equation}\nonumber
\begin{aligned}
    Dkq & \le 2f^*.\log(n)/\epsilon + 2f^*\log(\delta)/\epsilon + mf^*/n\\
    \Rightarrow k & \le 2\log(n)/\epsilon + 2\log(\delta)/\epsilon + m/n\\
    \Rightarrow E[x_a^m] & \le 2\log(n)/\epsilon + 2\log(\delta)/\epsilon + m/n(1 - q)\\
\end{aligned}
\end{equation}
  The probabilistic bound can be computed similar to the fixed $f$ case (Theorem~\ref{thm-fixed-f-beta}).
%Now, the $Pr[s_a > 4f\log(n)/\epsilon + 2f/\epsilon * \log(\delta)] \le
%Pr[ \Gamma(t) \ge nE[\Gamma(t)]] \le 1/n$ (using Markov's Inequality); where $s_a = \sum_{d=1}^{D} x_a^d$. Hence, $Pr[x_a^m > 2* [2f\log(n)/\epsilon + m/n(1 - f/D) + 2f/\epsilon * \log(\delta)] \le 1/n$. 
%
$\Box$
\end{proof}
%%%%
%%%%
%OLD => One can see that for, $m \le n\log(n)$, the gap becomes $O(\log(n)/\beta)$. Below, we give a matching lower bound, hence demonstrating that the upper bound is tight. Further, for even non-uniform distributions, when the probability of choosing a subset of $f$ populated dimensions is lower bounded by a constant $c$, the expected average load per dimension per bin is bounded by $mc/n$. Hence, $\Rightarrow E[x_a^m] \le 2\log(n)/\epsilon + 2\log(\delta)/\epsilon + mc/n(D/f - 1)$.
%%%%%%
%%%%%%%%%%% END OF LOOSE ANALYSIS %%%%%%%%%%%%%%%%%%%%%%%%%%
%%%%%%
%%%%%%
\subsection{Lower Bound}

We can show that the upper bound, for fixed $f$ case with uniform distribution, proved in section~\ref{sec:upper-bound-beta} is tight to within $f/D$ factor. Consider the case when, $an \log(n)/\beta^2$ balls are thrown into $n$ bins, using the $(1 + \beta)$ choice process. The expected dimensional sum load per bin is $af\log(n) / \beta^2$. Now, the expected number of balls thrown using the $(1+\beta)$ choice process is $an(1-\beta)\log(n)/\beta^2$. Raab and Steger~\cite{simple-analysis-bb} show that when $cn\log(n)$ balls are thrown uniformly and randomly into $n$ bins, then the load of the most loaded bin is at least $(c + \sqrt{c}/10)\log(n)$ balls. Using, $c = a(1-\beta)/\beta^2$, one can see that sum load in the maximum sum load bin is at least:
\begin{equation}\nonumber
\begin{aligned}
   & (\frac{a(1-\beta)}{\beta^2} + \sqrt{\frac{a(1-\beta)}{100\beta^2}} ) * f\log(n)
   & = (\frac{a}{\beta^2} + \frac{\sqrt{a(1-\beta)} - a}{10\beta} ) * f\log(n)
\end{aligned}
\end{equation}
Since, each ball has $f$ populated dimensions, hence, there are at least $O(\log(n)/\beta + a\log(n)/\beta^2)$ balls in this max sum load bin. Since, in each ball $f$ dimensions are uniformly distributed over $D$ dimensions, there exists a dimension whose load is at least $O(f\log(n)/D\beta)$ more than the average. Hence, the lower bound is $O(f\log(n)/D\beta)$.
\section{Conclusions \& Future Work}\label{sec:conc}
%%%%%%
In this paper, we consider the challenging problem of multidimensional balanced allocation for both the sequential and the parallel $d$ choice process and show that the gap (assuming fixed $f$ populated dimensions per ball and uniform distribution of $f$ over $D$) is $O(\log\log(n))$, which is tight (within $D/f$ factor of the lower bound). This improves the best prior~\cite{md-mm} bound of $O(\log\log(nD))$. Further, for arbitrary number of balls $m >> n$, the expected gap also has upper bound of $O(\log\log(n))$, that is independent of $m$ for the fixed $f$ case with uniform distribution of populated dimensions. For the variable $f$ case with (non-uniform) binomial distribution of populated dimensions, the gap is $O(\log(n))$ for $m=O(n)$. To the best of our knowledge, this is the first such analysis for $d$-choice paradigm with multidimensional balls and bins.
%%%%%
\par Our analysis also provides a much easier and elegant proof technique (as compared to~\cite{petra-heavy-case}) for the $O(\log\log(n))$ gap for $m >> n$ scalar balls thrown into $n$ bins using the symmetric multiple choice process. Moreover, for the weighted sequential scalar balls and bins and general case $m >> n$, we show the upper bound on the expected gap as $O(\log(n))$ which improves upon the best prior bound of  $n^c$ ($c$ depends on the weight distribution that has finite fourth moment) provided in~\cite{kunal-weighted}.In future, we would like to generalize the potential function approach for parallel and weighted balls and bins.
%Moreover, for $d$-choice process, we proved that for fixed $f$ per ball and uniform distribution of $f$ populated dimensions over $D$ total dimensions, the upper bound on the gap is $O(\log\log(n))$, which is better than the best prior~\cite{md-mm} bound of $O(\log\log(nD))$.
%%%%%%
%%%%%%
\par Further, we consider the challenging problem of multidimensional balanced allocation for the $(1+\beta)$ choice process and show that for arbitrarily large number of balls, the expected gap (assuming fixed $f$ populated dimensions per ball and uniform distribution of $f$ over $D$) is $O(\frac{\log(n)}{\beta})$, which is tight (within $D/f$ factor of the lower bound) and also independent of $m$. Further, the expected gap is also independent of $m$ for non-uniform distribution of $f$ dimensions over $D$, with fixed $f$ per ball) and for random $f$ with Binomial distribution.
%%%%%%
%%%%
\bibliographystyle{plain}
\bibliography{md-pbb.podc}

\appendix

\section{Visualization of Multidimensional balls and bins}
\label{app:fig}

\begin{figure}[h]
\centering
\includegraphics[width=3.5in]{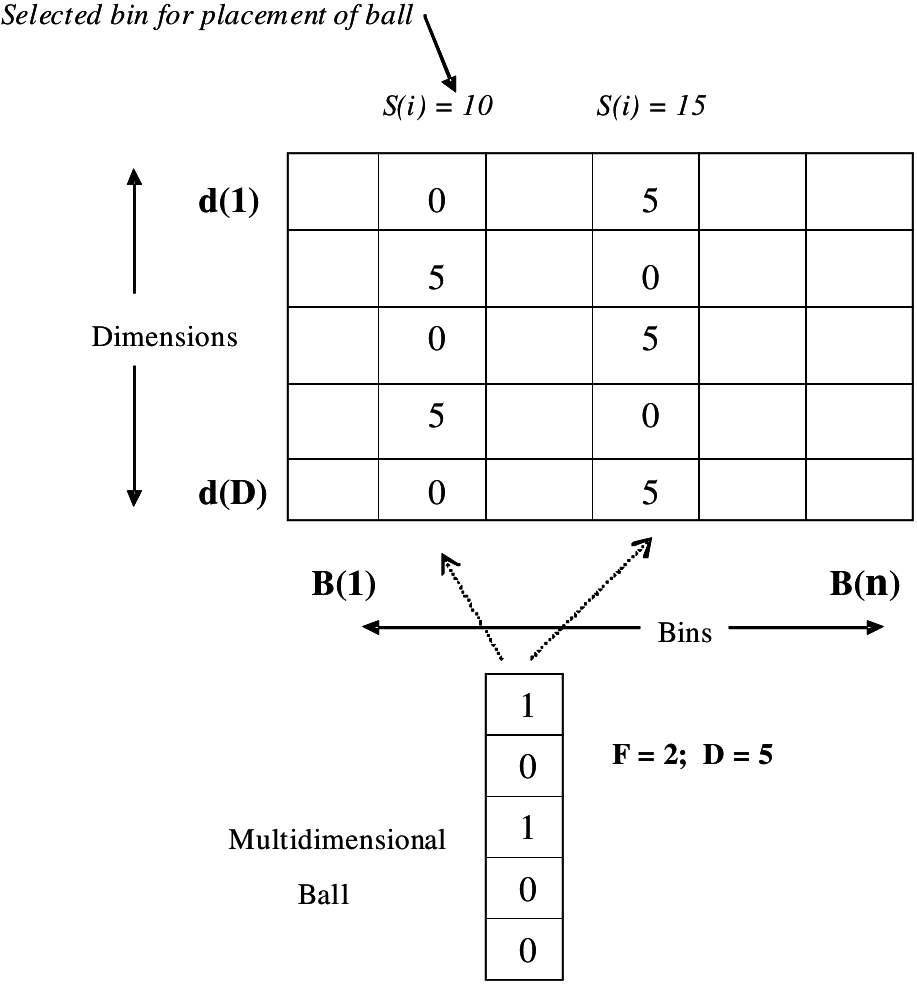}
\caption{Multidimensional Balls and Bins Scenario} 
\label{fig:md-balls-bins} 
\end{figure}

\section{Proof of Lemma~\ref{lemma-psi-base}} \label{app:1_proof}
The Lemma is restated below.
\begin{lemma}
  When an md-ball is thrown into an md-bin, the following inequality holds:
\begin{equation}\label{eq-psi-base}
   \EE[ \Psi(t+1) - \Psi(t) | x(t) ] \le
   \sum_{i=1}^{n} [ p_i * (-\alpha.f + (\alpha.f)^2)  + \alpha.f/n ]\Psi_i 
\end{equation}
\end{lemma}
\begin{proof}
  Let $\Lambda_i$ be the expected change in $\Psi$ if the ball is put in bin, $i$. So, $r_i(t+1) = s_i + f(1-1/n)$; and for $j \ne i$, $r_j(t+1) = s_j(t) - f/n$. The new values i.e. $s(t+1)$ are obtained by sorting $r(t+1)$ and $\Psi(s) = \Psi(r)$. The expected contribution of bin, $i$, to $\Lambda_i$ is given as follows:
\begin{equation}\nonumber
\begin{aligned}
   \EE[ \frac{e^{-\alpha.(s_i + f(1-1/n))}}{p_0} ] - \frac{e^{-\alpha.s_i}}{p_0}
   & = \frac{e^{-\alpha.s_i}}{p_0} [ e^{-\alpha.f(1-1/n)} - 1 ]
\end{aligned}
\end{equation}
 Similarly, the expected contribution of bin, $j$ ($j \ne i$) to $\Lambda_i$ is given as:
\begin{equation}\nonumber
\begin{aligned}
   \EE[ \frac{e^{-\alpha.(s_j - f/n)}}{j} ] - \frac{e^{-\alpha.s_j}}{j}
   & = \frac{e^{-\alpha.s_j}}{j} [ e^{\alpha.f/n} - 1 ]
\end{aligned}
\end{equation}
 Therefore, $\Lambda_i$ is given as follows:
\begin{equation}\nonumber
\begin{aligned}
   \Lambda_i & = \Psi_i [ e^{-\alpha.f(1-1/n)} - 1] +
              \sum_{j \ne i} \Psi_j(e^{\alpha.f/n} - 1)\\
   & = \Psi_i e^{\alpha f/n} (e^{-\alpha.f} - 1) + (e^{\alpha.f/n} - 1).\Psi
\end{aligned}
\end{equation}
  Thus, we get the overall expected change in $\Psi$ as follows:
\begin{equation}
\begin{aligned}
   \EE[ \Psi(t+1) - \Psi(t) | x(t) ] & = \sum_{i=1}^{n} p_i * \Lambda_i\\
   & = \sum_{i=1}^{n} p_i * [ \Psi_ie^{\alpha f/n} (e^{-\alpha.f} - 1) + (e^{-\alpha.f/n} - 1).\Psi ]\\
   & = \sum_{i=1}^{n} p_i * \Psi_ie^{\alpha f/n} (e^{-\alpha.f} - 1) +
     (e^{-\alpha.f/n} - 1).p_i.\Psi\\
   & = \sum_{i=1}^{n} [ p_i * e^{\alpha.f/n} (e^{-\alpha.f} - 1) + (e^{\alpha.f/n} - 1) ]\Psi_i
\end{aligned}
\end{equation}
 Now, $e^{\alpha.f/n)} (e^{-\alpha.f} - 1)$ can be approximated as follows:
\begin{equation}\nonumber
\begin{aligned}
   e^{\alpha.f/n)}.(e^{-\alpha.f} - 1)
   & \le (1 + \alpha.f/n + (\alpha.f/n)^2) * (1 - \alpha.f + (\alpha.f)^2 - 1)\\
   & \sim -\alpha.f + (\alpha.f)^2 + O((\alpha.f)^2/n)\\
   & \text{Since, $(\alpha.f)^2/n$ is very small for large $n$, we ignore these small terms. Hence,}\\
   e^{\alpha.f/n)}.(e^{-\alpha.f} - 1) &\lessapprox (-\alpha.f + (\alpha.f)^2)
\end{aligned}
\end{equation}
 Similarly, $(e^{\alpha.f/n} - 1) \lessapprox \alpha.f/n$
Hence, the expected change in $\Psi$ can be given by:
\begin{equation}\label{eq-psi-base}
\begin{aligned}
   \EE[ \Psi(t+1) - \Psi(t) | x(t) ]
   \le \sum_{i=1}^{n} [ p_i * (-\alpha.f + (\alpha.f)^2)  + \alpha.f/n ]\Psi_i
\end{aligned}
\end{equation}
%
%$\Box$
\end{proof}
%%%%%
%%%%%
%%%%%
%%%%%
\section{Proof of Lemma~\ref{lemma-psi-easy-case}} \label{app:2_proof}
%%%%%%
\begin{lemma}\label{app-lemma-psi-easy-case}
  Let $\Psi$ be defined as above. If $s_{(n\gamma_1)}(t) \ge 0$ then,
  $\EE[ \Psi(t+1) | x(t) ] \le (1-\frac{\alpha f}{8n})\Psi$ 
\end{lemma}
%%%%%%%
\begin{proof}
    From equation~\eqref{eq-psi-base}, we get,
\begin{equation}\label{eq-psi-1}
\begin{aligned}
  \EE[ \Psi(t+1) - \Psi(t) | x(t)] & \le \sum_{i=1}^{n} (p_i * (-\alpha.f + (\alpha.f)^2) + \alpha.f/n). \Psi_i\\
   & \le \sum_{i \ge n\gamma_1} p_i * (-\alpha.f + (\alpha.f)^2).\Psi_i + \alpha.f/n * \Psi
\end{aligned}
\end{equation}
   The last inequality follows since the $(-\alpha.f + (\alpha.f)^2)$ is negative.
  Now, we need to upper bound the term $\sum_{i \ge n\gamma_1} p_i * (-\alpha.f + (\alpha.f)^2).\Psi_i$. Since, $(-\alpha.f + (\alpha.f)^2)$ is negative, we need to find the minimum value of $\sum_{i \ge n\gamma_1} p_i * \Psi_i$. Further, $x_{n\gamma_1}(t) \ge 0$, so $\sum_{ i \ge n\gamma_1} \Psi_i \ge \Psi - \ln(n\gamma_1)$. Since $p_i$ is non-decreasing and $e^{-\alpha.x_i}$ is non-decreasing, the minimum value of $\sum_{i \ge n\gamma_1} p_i * \Psi_i$ is achieved, when, $e^{-\alpha s_i} \ln(1/\gamma_1) = (\Psi - \ln(n\gamma_1))$ for each $i \ge n\gamma_1$. Thus, the minimum value is given as follows.
%%%%%
\begin{equation}
\begin{aligned}
    \sum_{i \ge n\gamma_1} p_i\Psi_i & \ge \frac{\Psi - \ln(n\gamma_1)}{\ln(1/\gamma_1)} \sum_{i \ge n\gamma_1} \frac{2i-1}{n^2} * \frac{1}{i}\\
    & \ge \frac{\Psi - \ln(n\gamma_1)}{\ln(1/\gamma_1)} ( \frac{2(1-\gamma_1)}{n} - \frac{\ln(1/\gamma_1)}{n^2})
\end{aligned}
\end{equation}
%%%%
  Thus, the expected change in $\Psi$ can be computed, using equation~\eqref{eq-psi-1} and the above bound, as follows:
\begin{equation}
\begin{aligned}
  \EE[ \Psi(t+1) - \Psi(t) | x(t)] & \le (-\alpha.f + (\alpha.f)^2)* \frac{\Psi - \ln(n\gamma_1)}{\ln(1/\gamma_1)} ( \frac{2(1-\gamma_1)}{n} - \frac{\ln(1/\gamma_1)}{n^2}) + \frac{\alpha.f\Psi}{n}\\ 
   & \le \frac{-2\alpha f\Psi(1-\gamma_1)}{n\ln(1/\gamma_1)} - \frac{2\alpha.f(1-\gamma_1)\ln(n\gamma_1)}{\ln(1/\gamma_1)} + \frac{\alpha f\Psi}{n} + O((\alpha.f)^2) + O(1/n^2)\\
   & \le \frac{-\alpha f\Psi}{8n}
\end{aligned}
\end{equation}
%
%
%$\Box$\\
\end{proof}

\section{Proof of Lemma~\ref{lemma-s-gamma-psi}}
\label{app:3_proof}

\begin{proof}
    From equation~\eqref{eq-psi-base}, we get:
\begin{equation}
\begin{aligned}
   \EE[\Delta\Psi | x(t)] & \le \sum_{i=1}^{n} (p_i * (-\alpha.f + (\alpha.f)^2) + \alpha.f/n).\Psi_i\\
    & \le \sum_{i > (n\gamma_4)} (p_i * (-\alpha.f + (\alpha.f)^2) + \alpha.f/n).\Psi_i + \sum_{i \le (n\gamma_4)} (p_i * (-\alpha.f + (\alpha.f)^2) + \alpha.f/n) * \Psi_i\\
    & \le  (-\alpha.f + (\alpha.f)^2).\frac{\Psi_{> (n\gamma_4)}}{\ln(1/\gamma_4)} * \sum_{> (n\gamma_4)} \frac{2i-1}{n^2i} +\\
    & (-\alpha.f + (\alpha.f)^2).\frac{\Psi_{\le (n\gamma_4)}}{\ln(n\gamma_4)} * \sum_{\le (n\gamma_4)} \frac{2i-1}{n^2i} + \frac{\alpha f\Psi}{n}\\
    & \le \frac{\alpha f\Psi(-2(1-\gamma_4) + \ln(1\gamma_4))}{n\ln(1/\gamma_4)} + (2\alpha.f\Psi_{\le (n\gamma_4)})\frac{-\gamma_4\ln(1/\gamma_4) + (1-\gamma_4)\ln(n\gamma_4)}{n\ln(n\gamma_4)\ln(1/\gamma_4)}\\
\end{aligned}
\end{equation}
  In the above, the third inequality follows since, $(-\alpha.f + (\alpha.f)^2)$ is negative and $p_i \ge 0$. Now, since $\EE[\Delta\Psi|x(t)] \ge \frac{-\alpha f\Psi}{8n}$, we get that, %$\frac{\alpha f\Psi}{2n} \le \frac{2\alpha f\gamma_4}{n\ln(1/\gamma_4)}* \Psi_{\le (n\gamma_4)}$. Thus, we get: 
  $\Psi \le 4 \Psi_{\le (n\gamma_4)}$.
  Let, $B = \sum_{i} max(0, s_i)$ (as mentioned in Lemma~\ref{lemma-s-gamma-phi}). One can observe that, $\Psi_{\le (n\gamma_4)} \le \ln(n\gamma_4) * e^{ \alpha.B/(n\gamma_3)}$. This implies that, $\Psi \le 4\ln(n\gamma_4) e^{\alpha.B/(n\gamma_3)}$.

Since, $s_{n\gamma_1} < 0$, so, $\Phi \ge \ln(n\gamma_1) * e^{\alpha.B/(n\gamma_1)}$. If $\Psi < \epsilon\gamma_1 * \Phi$, then we are done. Else, $\Psi \ge \epsilon\gamma_1 * \Phi$. This implies:
\begin{equation}\nonumber
  4\ln(n\gamma_4)e^{\alpha.B/(n\gamma_3)} \ge \Psi \ge \epsilon\gamma_1 * \Phi \ge (\epsilon\gamma_1\ln(n\gamma_1)) * e^{\alpha B/(n\gamma_1)}
\end{equation}
  Thus, $e^{\alpha.B/n} \le (\frac{4\ln(n\gamma_4)}{\epsilon\gamma_1\ln(n\gamma_1)})^\frac{\gamma_1\gamma_3}{\gamma_3 - \gamma_1}$. So, $\Gamma \le ((\epsilon+\theta)/\epsilon) * \Psi \le ((\epsilon+\theta)/\epsilon) * 4\ln(n\gamma_4)e^{\alpha.B/(n\gamma_3)}$. Hence, $\Gamma \le c\ln(n)$, where, $c = poly(1/\epsilon)$.
%
%$\Box$\\
\end{proof}
%%%%%%%%
%%%%%%%%
%%%%%%%%%%%%%%%%%%%%%%%%%%%%%%%%%%%%%%%%%%%%%%%%%%%%%%%%%%%%%
%%%%%%%%%%%%%%%%%% Proofs for Weighted Case %%%%%%%%%%%%%%%%%%%%%%%%%%%%%
%%%%%%%%%%%%%%%%%%%%%%%%%%%%%%%%%%%%%%%%%%%%%%%%%%%%%%%%%%%%%
%%%%%%%%
%%%%%%%%
\section{Proof of Lemma~\ref{lemma-psi-base-w}} \label{app:weighted_1_proof}
The Lemma is restated below.
\begin{lemma}
  When an md-ball is thrown into an md-bin, the following inequality holds:
\begin{equation}\label{eq-psi-base-w}
   \EE[ \Psi(t+1) - \Psi(t) | x(t) ] \le
   \sum_{i=1}^{n} [ p_i * (-\alpha.f^* + \frac{S\alpha^2}{n^2})  + \alpha.f^*/n ]\Psi_i 
\end{equation}
\end{lemma}
\begin{proof}
  Let $\Lambda_i$ be the expected change in $\Psi$ if the ball is put in bin, $i$. So, $r_i(t+1) = s_i + f(1-1/n)$; and for $j \ne i$, $r_j(t+1) = s_j(t) - f/n$. The new values i.e. $s(t+1)$ are obtained by sorting $r(t+1)$ and $\Psi(s) = \Psi(r)$. When, an md-ball is committed to bin $i$, then it jumps to an index $i_{new}$ which is less than or equal to $i$ in the new bin order. Using similar analysis for taking care of these jumps as in Lemma~\ref{lemma-phi-base-w}, the expected contribution of bin, $i$, to $\Lambda_i$ is given as follows:
\begin{equation}\nonumber
\begin{aligned}
   \EE[ \frac{e^{-\alpha.(s_i + f(1-1/n))}}{n^2 + n-i+1} ] - \frac{e^{-\alpha.s_i}}{n^2+n-i+1}\\
   & = \frac{e^{-\alpha.s_i}}{n^2+n-i+1} [ M(-\alpha(1-1/n)) - 1 ]\\
   &\le \Psi_i (-\alpha.f^*(1-1/n) + \frac{S\alpha^2}{n^2})
\end{aligned}
\end{equation}
 Similarly, the expected contribution of bin, $j$ ($j \ne i$) to $\Lambda_i$ is given as:
\begin{equation}\nonumber
\begin{aligned}
   \EE[ \frac{e^{-\alpha.(s_j - f/n)}}{n^2+n-j+1} ] - \frac{e^{-\alpha.s_j}}{n^2+n-j+1}\\
   & = \frac{e^{-\alpha.s_j}}{n^2+n-j+1} [ M(\alpha/n) - 1 ]\\
   &\le \Psi_j \frac{-\alpha.f^*}{n}
\end{aligned}
\end{equation}
 Therefore, $\Lambda_i$ is given as follows:
\begin{equation}\nonumber
\begin{aligned}
   \Lambda_i & = \Psi_i [\alpha.f^* + \frac{S\alpha^2}{n^2}] + \frac{\alpha.f^*\Psi}{n}
%   & = \Psi_i e^{\alpha f/n} (e^{-\alpha.f} - 1) + (e^{\alpha.f/n} - 1).\Psi
\end{aligned}
\end{equation}
%
%
%  Thus, we get the overall expected change in $\Psi$ as follows:
%\begin{equation}
%\begin{aligned}
%   \EE[ \Psi(t+1) - \Psi(t) | x(t) ] & = \sum_{i=1}^{n} p_i * \Lambda_i\\
%   & = \sum_{i=1}^{n} p_i * [ \Psi_ie^{\alpha f/n} (e^{-\alpha.f} - 1) + (e^{-\alpha.f/n} - 1).\Psi ]\\
%   & = \sum_{i=1}^{n} p_i * \Psi_ie^{\alpha f/n} (e^{-\alpha.f} - 1) +
%     (e^{-\alpha.f/n} - 1).p_i.\Psi\\
%   & = \sum_{i=1}^{n} [ p_i * e^{\alpha.f/n} (e^{-\alpha.f} - 1) + (e^{\alpha.f/n} - 1) ]\Psi_i
%\end{aligned}
%\end{equation}
%
% Now, $e^{\alpha.f/n)} (e^{-\alpha.f} - 1)$ can be approximated as follows:
%\begin{equation}\nonumber
%\begin{aligned}
%   e^{\alpha.f/n)}.(e^{-\alpha.f} - 1)
%   & \le (1 + \alpha.f/n + (\alpha.f/n)^2) * (1 - \alpha.f + (\alpha.f)^2 - 1)\\
%   & \sim -\alpha.f + (\alpha.f)^2 + O((\alpha.f)^2/n)\\
%   & \text{Since, $(\alpha.f)^2/n$ is very small for large $n$, we ignore these small terms. Hence,}\\
%   e^{\alpha.f/n)}.(e^{-\alpha.f} - 1) &\lessapprox (-\alpha.f + (\alpha.f)^2)
%\end{aligned}
%\end{equation}
%
% Similarly, $(e^{\alpha.f/n} - 1) \lessapprox \alpha.f/n$
%
Hence, the expected change in $\Psi$ can be given by:
\begin{equation}
\begin{aligned}
   \EE[ \Psi(t+1) - \Psi(t) | x(t) ]
   \le \sum_{i=1}^{n} [ p_i * (-\alpha.f^* + \frac{S\alpha^2}{n^2}) + \alpha.f^*/n ]\Psi_i
\end{aligned}
\end{equation}
%
%$\Box$
\end{proof}
%%%%
%%%%
%%%%
\section{Proof of Lemma~\ref{lemma-psi-easy-case-w}} \label{app:2_w_proof}
%%%%%%
\begin{lemma}\label{app-lemma-psi-easy-case-w}
  Let $\Psi$ be defined as above. If $s_{(n\gamma_1)}(t) \ge 0$ then,
  $\EE[ \Psi(t+1) | x(t) ] \le (1-\frac{\alpha f^*}{2n})\Psi$ 
\end{lemma}
%%%%%%%
\begin{proof}
    From equation~\eqref{eq-psi-base-w}, we get,
\begin{equation}\label{eq-psi-w-1}
\begin{aligned}
  \EE[ \Psi(t+1) - \Psi(t) | x(t)] & \le \sum_{i=1}^{n} (p_i * (-\alpha.f^* + \frac{S(\alpha)^2}{n^2}) + \alpha.f^*/n). \Psi_i\\
   & \le \sum_{i \ge n\gamma_1} p_i * (-\alpha.f^* + \frac{S\alpha^2}{n^2}).\Psi_i + \alpha.f^*/n * \Psi
\end{aligned}
\end{equation}
   The last inequality follows since the $(-\alpha.f^* + \frac{S\alpha^2}{n^2})$ is negative.
  Now, we need to upper bound the term $\sum_{i \ge n\gamma_1} p_i * (-\alpha.f^* + \frac{S\alpha^2}{n^2}).\Psi_i$. Since, $(-\alpha.f^* + \frac{S\alpha^2}{n^2})$ is negative, we need to find the minimum value of $\sum_{i \ge n\gamma_1} p_i * \Psi_i$. %Further, $x_{n\gamma_1}(t) \ge 0$, so $\sum_{ i \ge n\gamma_1} \Psi_i \ge \Psi - \ln(1/\gamma_1)$. 
Since $p_i$ is non-decreasing and $\Psi_i$ is non-decreasing, the minimum value of $\sum_{i \ge n\gamma_1} p_i * \Psi_i$ is achieved, when, $e^{-\alpha s_i} \frac{n-n\gamma_1}{n^2+1} = \Psi$ for each $i \ge n\gamma_1$. Thus, the minimum value is given as follows.
%%%%%
\begin{equation}
\begin{aligned}
    \sum_{i \ge n\gamma_1} p_i\Psi_i & \ge \frac{n\Psi}{(1-\gamma_1)} \sum_{i \ge n\gamma_1} \frac{2i-1}{n^2} * \frac{1}{n^2+n-i+1}\\
    & \ge \frac{\Psi}{n(1 - \gamma_1)} * \frac{(2n+1)(1-\gamma_1)}{n}
\end{aligned}
\end{equation}
%%%%
  Thus, the expected change in $\Psi$ can be computed, using equation~\eqref{eq-psi-w-1} and the above bound, as follows:
\begin{equation}
\begin{aligned}
  \EE[ \Psi(t+1) - \Psi(t) | x(t)] & \le (-\alpha.f^* + \frac{S\alpha^2}{n^2})* \frac{\Psi}{n(1 - \gamma_1)}* \frac{(2n+1)(1-\gamma_1)}{n} + \frac{\alpha.f^*\Psi}{n}\\ 
   & \le \frac{-\alpha f^*\Psi}{n} + \frac{2S\alpha^2\Psi}{n^2} - \frac{\alpha.f^*\Psi}{n^2}\\
   & \le \frac{-\alpha.f^*\Psi}{2n}
\end{aligned}
\end{equation}
%
%$\Box$\\
\end{proof}
%%%%%
%%%%%
%%%%%
\section{Proof of Lemma~\ref{lemma-s-gamma-psi-w}}\label{app:3_w_proof}
%%%%%
%%%%%
 The lemma is restated below:
\begin{lemma} \label{app:lemma-s-gamma-psi-w}
   Let, $s_{(n\gamma_1)} < 0$ and $\EE[\Delta\Psi|x(t)] \ge -\alpha.f^*\Psi/2n$. Then, either $\Psi < \epsilon\gamma_1 * \Phi$, or $\Gamma < c/n$ for some $c = poly(1/epsilon)$.
\end{lemma}
%%%%%
%%%%%
\begin{proof}
    From equation~\eqref{eq-psi-base-w}, we get:
\begin{equation}
\begin{aligned}
   \EE[\Delta\Psi | x(t)] & \le \sum_{i=1}^{n} (p_i * (-\alpha.f^* + \frac{S\alpha^2}{n^2} + \frac{\alpha.f^*}{n}).\Psi_i\\
%%%%
    & \le \sum_{i > (n\gamma_4)} (p_i * (-\alpha.f^* + \frac{S\alpha^2}{n^2}) + \alpha.f^*/n).\Psi_i + \sum_{i \le (n\gamma_4)} (p_i * (-\alpha.f^* + \frac{S\alpha^2}{n^2}) + \alpha.f^*/n) * \Psi_i\\
%%%%
    & \le \frac{\alpha.f^*\Psi_{(\ge n\gamma_4)}}{n\gamma_4}*\frac{(2n+1)(\gamma_4-1)}{n} + \frac{\alpha.f^*}{n}\\
    & \frac{\alpha.f^*\Psi_{(< n\gamma_4)}}{n(\gamma_4-1)}*\frac{(2n+1)\gamma_4}{n}\\
%%%%
    &\le \frac{\alpha.f^*\Psi}{n\gamma_4}*\frac{(2n+1)(\gamma_4-1)}{n} + \frac{\alpha.f^*}{n}\\
    & \frac{2\alpha.f^*\Psi_{(< n\gamma_4)}}{n\gamma_4(\gamma_4-1)}*\frac{(1-2\gamma_4)}{(1-\gamma_4)}\\
\end{aligned}
\end{equation}
  In the above, the third inequality follows since, $(-\alpha.f^* + \frac{S\alpha^2}{n^2})$ is negative and $p_i \ge 0$. Now, since $\EE[\Delta\Psi|x(t)] \ge \frac{-\alpha.f^*\Psi}{2n}$, we get: $\Psi \le \Psi_{\le (n\gamma_4)} *  \frac{4(2\gamma_4-1)}{\gamma_4(1-\gamma_4)(4\gamma_4-1)}$.
  Let, $B = \sum_{i} max(0, s_i)$ (as mentioned in Lemma~\ref{lemma-s-gamma-phi-w}). One can observe that, $\Psi_{\le (n\gamma_4)} \le  * \frac{1-\gamma_4}{n} e^{\alpha.B/(n-n\gamma_4)}$. This implies that, $\Psi \le \frac{4(2\gamma_4-1)}{\gamma_4(4\gamma_4-1)} e^{\alpha.B/(n-n\gamma_4)}$.

Since, $s_{n\gamma_1} < 0$, so, $\Phi \ge \frac{\gamma_1}{n} * e^{\alpha.B/(n\gamma_1)}$. If $\Psi < \epsilon\gamma_1 * \Phi$, then we are done. Else, $\Psi \ge \epsilon\gamma_1 * \Phi$. This implies:
\begin{equation}\nonumber
  \frac{4(2\gamma_4-1)}{\gamma_4(4\gamma_4-1)}.e^{\alpha.B/(n-n\gamma_4)} \ge \Psi \ge \epsilon\gamma_1 * \Phi \ge \frac{\epsilon\gamma_1^2}{n} * e^{\alpha B/(n\gamma_1)}
\end{equation}
  Thus, $e^{\alpha.B/n} \le (\frac{4(2\gamma_4-1)}{\epsilon\gamma_1^2\gamma_4(4\gamma_4-1)})^\frac{\gamma_1(1-\gamma_4)}{1 - \gamma_4 - \gamma_1}$. So, $\Gamma \le ((1+\theta)/\epsilon) * \Psi \le ((1+\theta)/\epsilon) * \frac{4(2\gamma_4-1)}{n\gamma_4(4\gamma_4-1)} e^{\alpha.B/(n\gamma_3)}$. Hence, $\Gamma \le \frac{c}{n}$, where, $c = poly(1/\epsilon)$.
%
%$\Box$\\
\end{proof}
%%%%%
%%%%%
%%%%%%%%%%%% PROOFS FOR Psi: (1+beta) Unweighted Case %%%%%%%%%%%%%%%%%
%%%%%%%%%
%%%%%
\section{Proof of Lemma~\ref{lemma-psi-easy-case-beta}}
\label{app:1_proof_beta}
%%%%%%
\begin{lemma}\label{app-lemma-psi-easy-case-beta}
  Let $\Psi$ be defined as above. If $s_{(n\gamma_1)}(t) \ge 0$ then,
  $\EE[ \Psi(t+1) | x(t) ] \le (1-\frac{\epsilon^2}{4n})\Psi + 1$ 
\end{lemma}
%%%%%%%
\begin{proof}
    From equation~\eqref{eq-psi-base-beta}, we get,
\begin{equation}\label{eq-psi-1-beta}
\begin{aligned}
  \EE[ \Psi(t+1) - \Psi(t) | x(t)] & \le \sum_{i=1}^{n} (p_i * (-\alpha.f + (\alpha.f)^2) + \alpha.f/n). e^{-\alpha.s_i}\\
   & \le \sum_{i \ge n\gamma_1} p_i * (-\alpha.f + (\alpha.f)^2).e^{-\alpha.s_i} + \alpha.f/n * \Psi
\end{aligned}
\end{equation}
   The last inequality follows since the $(-\alpha.f + (\alpha.f)^2)$ is negative.
  Now, we need to upper bound the term $\sum_{i \ge n\gamma_1} p_i * (-\alpha.f + (\alpha.f)^2).e^{-\alpha.s_i}$. Since, $(-\alpha.f + (\alpha.f)^2)$ is negative, we need to find the minimum value of $\sum_{i \ge n\gamma_1} p_i * e^{-\alpha.s_i}$. Further, $x_{n\gamma_1}(t) \ge 0$, so $\sum_{ i \ge n\gamma_1} \Psi_i \ge \Psi - n\gamma_1$. Since $p_i$ is non-decreasing and $\Psi_i$ is non-decreasing, the minimum value of $\sum_{i \ge n\gamma_1} p_i * e^{-\alpha.s_i}$ is achieved, when, $\Psi_i = (\Psi - n\gamma_1)/n\gamma_2$ for each $i \ge n\gamma_1$. Using the assumption, that $\sum_{i \ge n\gamma_1} p_i \ge (\gamma_2 + \epsilon)$, the minimum value is: $(\gamma_2 + \epsilon)(\Psi - n\gamma_1)/(n\gamma_2)$. Thus,
\begin{equation}
   \sum_{i \ge n\gamma_1} p_i * (-\alpha.f + (\alpha.f)^2).e^{-\alpha.s_i} \le (-\alpha.f + (\alpha.f)^2)* (\gamma_2 + \epsilon)* (\Psi - n\gamma_1)/(n\gamma_2)
\end{equation}
  Thus, the expected change in $\Psi$ can be computed, using equation~\eqref{eq-psi-1-beta} and the above bound, as follows:
\begin{equation}
\begin{aligned}
  \EE[ \Psi(t+1) - \Psi(t) | x(t)] & \le (-\alpha.f + (\alpha.f)^2)(\gamma_2 + \epsilon)(\Psi - n\gamma_1)/(n\gamma_2) + \alpha.f/n * \Psi\\ 
   & \le \frac{-\epsilon^2}{2n\gamma_2}\Psi + \frac{\epsilon^2}{4n}\Psi + 1\\
   & \le \frac{-\epsilon^2}{4n}\Psi + 1
\end{aligned}
\end{equation}
%
%
%$\Box$\\
\end{proof}

\section{Proof of Lemma~\ref{lemma-s-gamma-psi-beta}}
\label{app:2_proof_beta}

\begin{proof}
    From equation~\eqref{eq-psi-base-beta}, we get:
\begin{equation}
\begin{aligned}
   \EE[\Delta\Psi | x(t)] & \le \sum_{i=1}^{n} (p_i * (-\alpha.f + (\alpha.f)^2) + \alpha.f/n).e^{-\alpha.s_i}\\
    & \le \sum_{i > (n\gamma_4)} (p_i * (-\alpha.f + (\alpha.f)^2) + \alpha.f/n).e^{-\alpha.s_i} + \sum_{i \le (n\gamma_4)} (p_i * (-\alpha.f + (\alpha.f)^2) + \alpha.f/n) * e^{-\alpha.s_i}\\
    & \le  [((1 + \theta\epsilon)/n) * (-\alpha.f + (\alpha.f)^2) + \alpha.f/n)].\Psi_{> (n\gamma_4)} + (\alpha.f).\Psi_{\le n\gamma_4}/n\\
    & \le [-\alpha f\theta\epsilon^2/n + \epsilon^2/4n].\Psi_{> (n\gamma_4)} + \alpha.f.\Phi_{\le (n\gamma_4)}/n\\
    & \le (-\epsilon^2/(2n\gamma_1) + \epsilon^2/(4n)) * \Psi + ( \alpha.f/n + \epsilon^2/(2n\gamma_1)).\Psi_{\le (n\gamma_4)}
\end{aligned}
\end{equation}
  In the above, the third inequality follows since, $(-\alpha.f + (\alpha.f)^2)$ is negative and $p_i \ge 0$; and $\forall i > (n\gamma_4), p_i > (1+\theta\epsilon)/n$. Now, since $\EE[\Delta\Psi|x(t)] \ge -
  \epsilon^2\Psi/4n$, we get that, $\epsilon^2\gamma_2\Psi/(2n\gamma_1) \le (\epsilon/2n + \epsilon^2/(2n\gamma_1))* \Psi_{\le (n\gamma_4)}$. Thus, we get: $\Psi \le \frac{(1+\theta\epsilon)}{\epsilon\gamma_2} * \Psi_{\le (n\gamma_4)}$.
  Let, $B = \sum_{i} max(0, s_i)$ (as mentioned in Lemma~\ref{lemma-s-gamma-phi-beta}). One can observe that, $\Psi_{\le (n\gamma_4)} \le (n\gamma_4) * e^{ \alpha.B/(n\gamma_3)}$. This implies that, $\Psi \le \frac{(1+\theta\epsilon)}{\epsilon\gamma_2} * (n\gamma_4) * e^{\alpha.B/(n\gamma_3)}$.

Since, $s_{n\gamma_1} < 0$, so, $\Phi \ge n\gamma_1 * e^{\alpha.B/(n\gamma_1)}$. If $\Psi < \epsilon\gamma_1 * \Phi$, then we are done. Else, $\Psi \ge \epsilon\gamma_1 * \Phi$. This implies:
\begin{equation}\nonumber
  ((n(1+\theta\epsilon)\gamma_4)/(\epsilon\gamma_2)).e^{\alpha.B/(n\gamma_3)} \ge \Psi \ge \epsilon\gamma_1 * \Phi \ge (\epsilon\gamma_1^2n) * e^{\alpha B/(n\gamma_1)}
\end{equation}
  Thus, $e^{\alpha.B/n} \le (\frac{(1+\theta\epsilon)\gamma_4}{\epsilon^2\gamma_2\gamma_1^2})^\frac{\gamma_1\gamma_3}{\gamma_3 - \gamma_1}$. So, $\Gamma \le ((1+\theta)/\epsilon) * \Psi \le ((1+\theta)/\epsilon) * \frac{(1+\theta\epsilon)\gamma_4n}{\epsilon\gamma_2}.e^{\alpha.B/(n\gamma_3)}$. Hence, $\Gamma \le cn$, where, $c = poly(1/\epsilon)$.
%
%$\Box$\\
\end{proof}
%%%%%
%%%%%
\end{document}